\documentclass[10pt]{amsart}
\usepackage{times,amsmath,amsbsy,amssymb,amscd,mathrsfs}
\usepackage{graphicx,subfigure,epstopdf,wrapfig,chemarrow}

\usepackage{algorithm2e} 
\usepackage{multicol,multirow}
\usepackage{mathtools}
\usepackage[usenames,dvipsnames,svgnames,table]{xcolor}
\usepackage[all]{xy}
\usepackage{wrapfig}
\usepackage{tcolorbox}

\usepackage{tikz,tikz-cd}
\usepackage[utf8]{inputenc}
\usepackage{pgfplots} 
\usepackage{pgfgantt}
\usepackage{pdflscape}
\pgfplotsset{compat=newest} 
\pgfplotsset{plot coordinates/math parser=false}
\newlength\fwidth

\definecolor{myBlue}{rgb}{0.0,0.0,0.55}
\usepackage[pdftex,colorlinks=true,citecolor=myBlue,linkcolor=myBlue]{hyperref}

\usepackage[hyperpageref]{backref}

\usepackage{comment,enumerate,multicol,xspace}

  \newcounter{mnote}
  \let\oldmarginpar\marginpar
    \renewcommand\marginpar[1]{\-\oldmarginpar[\raggedleft\footnotesize #1]%
    {\raggedright\footnotesize #1}}

\newtheorem{theorem}{Theorem}[section]
\newtheorem{lemma}[theorem]{Lemma}

\newtheorem{remark}[theorem]{Remark}

\newcommand{\dd}{\,{\rm d}}

\DeclareMathOperator*{\img}{img}

\newcommand{\curl}{{\rm curl\,}}
\renewcommand{\div}{\operatorname{div}}
\newcommand{\grad}{{\rm grad\,}}
\newcommand{\rot}{{\rm rot}}

\newcommand{\tr}{\operatorname{tr}}
\newcommand{\dev}{\operatorname{dev}}
\newcommand{\sym}{\operatorname{sym}}
\newcommand{\skw}{\operatorname{skw}}

\newcommand{\mskw}{\operatorname{mskw}}
\newcommand{\vskw}{\operatorname{vskw}}

\newcommand{\hess}{\operatorname{hess}}

\newcommand{\vertiii}[1]{{\left\vert\kern-0.25ex\left\vert\kern-0.25ex\left\vert #1 
    \right\vert\kern-0.25ex\right\vert\kern-0.25ex\right\vert}}

\begin{document}
\title[Discrete Hessian Complexes in Three Dimensions]{Discrete Hessian Complexes in Three Dimensions}
\author{Long Chen}%
 \address{Department of Mathematics, University of California at Irvine, Irvine, CA 92697, USA}%
 \email{chenlong@math.uci.edu}%
 \author{Xuehai Huang}%
 \address{Corresponding author. School of Mathematics, Shanghai University of Finance and Economics, Shanghai 200433, China}%
 \email{huang.xuehai@sufe.edu.cn}%

 \thanks{Long Chen was supported by the National Science Foundation (NSF) DMS-1913080 and DMS-2012465. Xuehai Huang was supported by the National Natural Science Foundation of China Project 11771338 and 12071289, the Natural Science Foundation of Shanghai 21ZR1480500, and the Fundamental Research Funds for the Central Universities 2019110066.}


\maketitle


\begin{abstract}
A family of conforming virtual element Hessian complexes on tetrahedral meshes are constructed based on decompositions of polynomial tensor spaces.  They are applied to discretize the linearized time-independent Einstein-Bianchi system with optimal order convergence.
\end{abstract}

\section{Introduction}
\label{sec:1}
Let $\Omega$ be a bounded Lipschitz domain in $\mathbb R^3$. The Hessian complex, also known as grad-grad complex, in three dimensions reads as~\cite{ArnoldHu2020,PaulyZulehner2020}
\begin{equation*}
\mathbb P_1(\Omega)\xrightarrow{\subset} H^{2}(\Omega)\xrightarrow{\hess}\boldsymbol H(\curl, \Omega;\mathbb S)\xrightarrow{\curl} \boldsymbol H(\div, \Omega;\mathbb T) \xrightarrow{{\div}} \boldsymbol L^2(\Omega;\mathbb R^3)\xrightarrow{}\boldsymbol0,
\end{equation*}
where $\mathbb P_1(\Omega)$ is the linear polynomial space, $H^2(\Omega)$ and $\boldsymbol L^2(\Omega;\mathbb R^3)$ are standard Sobolev spaces,  $\boldsymbol H(\curl, \Omega;\mathbb S)$ is the space of symmetric matrices whose $\curl$ is in $\boldsymbol L^2(\Omega;\mathbb T)$, and $\boldsymbol H(\div, \Omega;\mathbb T)$ is the space of trace-free matrices whose $\div$ is in $\boldsymbol L^2(\Omega;\mathbb R^3)$. 
Here both $\curl$ and $\div$ are applied to matrices row-wisely.
Given a tetrahedral mesh of domain $\Omega$, we shall construct discrete Hessian complexes with conforming virtual element spaces and apply to solve the linearized Einstein-Bianchi (EB) system~\cite{Quenneville-Belair2015}. 

Finding finite elements with continuous derivatives (the so-called $C^1$ element), symmetry, or trace-free leads to higher number of degrees of freedom. To avoid this issue, Arnold and Quenneville-Belair~\cite{Quenneville-Belair2015} use multipliers to impose the weak $H^2$-conforming and weak symmetry and obtain an optimal order discretization of the EB system. In~\cite{HuLiang2020} Hu and Liang construct the first finite element Hessian complexes in three dimensions. The lowest order complex starts with the $\mathbb P_9$ $\mathcal C^1$-element constructed in~\cite{Zhang:2009family} and consists of $\mathbb P_7$ for $\boldsymbol H(\curl,\Omega; \mathbb S)$ and $\mathbb P_6$ for $\boldsymbol H(\div, \Omega;\mathbb T)$, where $\mathbb P_k$ stands for the polynomail space of degree $k$. Although the practical significance may be limited due to the high polynomial degree of the elements, the work~\cite{HuLiang2020} is the first construction of conforming discrete Hessian complexes consisting of finite element spaces in $\mathbb R^3$, and it motivates us to the development of simpler methods.

We shall use ideas of virtual element methods (VEMs) to construct discrete Hessian complexes with fewer degrees of freedom. The virtual element developed in ~\cite{BeiraoBrezziCangianiManziniEtAl2013,BeiraoBrezziMariniRusso2014}  is a generalization of the finite element on tensorial/simplicial meshes to a general polyhedral mesh and can be also thought of as a variational framework for the mimetic finite difference methods~\cite{Brezzi;Lipnikov;Simoncini:2005family,Lipnikov;Manzini;Shashkov:2014Mimetic}. Compared with the standard finite element methods mainly working on tensorial/simplicial meshes, VEMs have a variety of distinct advantages. The VEMs are, foremost, highly adaptable to the polygonal/polyhedral, and even anisotropic quadrilateral/hexahedral meshes. For problems with complex geometries, this leads to great convenience in the mesh generation, e.g., discrete fracture network 
simulations~\cite{Benedetto;Berrone;Pieraccini;Scialo:2014virtual}, and the elliptic interface problems in three dimensions~\cite{Chen;Wei;Wen:2017interface-fitted}. Another trait of VEMs is its astoundingly painless construction of smooth elements for high-order elliptic problems. For instance, $H^2$-conforming VEMs have been constructed in ~\cite{BrezziMarini2013,Antonietti;Da-Veiga;Scacchi;Verani:2016Virtual,BeiraodaVeigaDassiRusso2020} which shows a simple and elegant construction readily to be implemented. A uniform construction of the $H^m$-nonconforming virtual elements of any order $k$ and $m$ on any shape of polytopes in $\mathbb R^n$ with constraint $k\geq m$ has been developed in~\cite{Chen;Huang:2020Nonconforming,Huang:2020Nonconforming}. 
One more merit is that the virtual element space can be devised to be structure preserving, such as the harmonic VEM~\cite{ChernovMascotto2019,MascottoPerugiaPichler2018} and the divergence-free Stokes VEMs~\cite{BeiraodaVeigaLovadinaVacca2017,WeiHuangLi2020}.
VEMs for de Rham complex~\cite{Da-Veiga;Brezzi;Marini;Russo:2016curl-conforming} and Stokes complexes~\cite{Beirao-da-Veiga;Dassi;Vacca:2020Stokes} have been also constructed recently. 

In the construction of the VEM spaces, the subtlest and a key component is the well-posedness of a local problem with non-zero Dirichlet boundary conditions. Take an $H^2$-conforming VEM space as an example. 
Given data $(f,g_1,g_2)$, consider the biharmonic equation with Dirichlet boundary condition on a polyhedron $K$
\begin{equation}\label{intro:biharmonic1}
\Delta^2 v = f \text{ in } K, \quad v = g_1, \partial_n v = g_2 \text{ on } \partial K.
\end{equation}
When $g_1 = g_2 = 0$, the existence and uniqueness is a consequence of the Lax-Milligram lemma on $H_0^2(K)$. The classical way to deal with the non-zero Dirichlet boundary condition $(g_1, g_2)$ is to find a lifting $v^b\in H^2(K)$ with $v^b = g_1, \partial_n v^b = g_2$ and change~\eqref{intro:biharmonic1} to the homogenous boundary condition with modified source $f - \Delta^2 v^b$. Such lifting is guaranteed by trace theorems of Sobolev spaces which is usually established for smooth domains. For polyhedral domains, however, compatible conditions~\cite{BuffaGeymonat2001} are needed. Although the traces $g_1$ and $g_2$ are defined piece-wisely on each face  $F$ of $K$, for $H^2$-functions, $\left(g_2|_F\boldsymbol n_F+\nabla_F(g_1|_F)\right)|_e$ should be single-valued across each edge $e$ of the polyhedron $K$, for $F$ containing edge $e$. That is $g_1$ and $g_2$ cannot be chosen independently.

For vector function spaces, characterization of the trace spaces and corresponding compatible conditions is harder as tangential and normal components of the trace should be treated differently. We refer to~\cite[Appendix A]{Beirao-da-Veiga;Dassi;Vacca:2020Stokes} for the discussion of the well-posedness of the biharmonic problem of vector functions with a non-homogeneous boundary conditions, and refer to~\cite{Buffa:2003theorems} and references therein for the trace of $H(\curl,\Omega)$, where variants of space $H^{1/2}(\partial\Omega)$ are introduced. Specifically a lifting for the trace of a function in $H(\curl,\Omega)$ on a Lipschitz domain is explicitly constructed in~\cite{Tartar:1997characterization} which is highly non-trivial. 

We are not able to characterize the trace space of $\boldsymbol H(\curl,\Omega; \mathbb S)$ and thus cannot follow the classical approach of VEM to define the shape function space using local problems. Instead we still consider tetrahedron element $K$ and combine finite element and virtual element spaces. We first establish a polynomial Hessian complex and corresponding Koszul complex, which leads to the  decomposition 
$$
\mathbb P_{k}(K;\mathbb S) =\nabla^2 \, \mathbb P_{k+2}(K)\oplus\sym(\mathbb P_{k-1}(K;\mathbb T)\times\boldsymbol x) \quad k\geq 1. 
$$
Based on this decomposition, we can construct a virtual element space
$$
\boldsymbol\Sigma(K)=\nabla^2W(K)\oplus \sym(\boldsymbol V(K)\times\boldsymbol x),
$$
where $W(K)$ is an $H^2$-conforming VEM space and $\boldsymbol V(K) = \mathbb P_{k-1}(K;\mathbb T)$ is an $H(\div)$-conforming finite element space. Degrees of freedom for $\boldsymbol \Sigma(K)$ are carefully chosen so that the resulting global space $\boldsymbol \Sigma_h$ is $H(\curl)$-conforming and its $L^2$-projection to $\mathbb P_k(\mathbb S)$ is computable. Our construction is different from the approach in \cite{HuLiang2020} for constructing a finite element Hessian complex, where characterization of polynomial bubble functions is crucial. 

Our $H^2$-conforming virtual element $W(K)$ is slightly different from those constructed in~\cite{BrezziMarini2013,Antonietti;Da-Veiga;Scacchi;Verani:2016Virtual,BeiraodaVeigaDassiRusso2020}. Again we take the advantage of $K$ being a tetrahedron to construct an element so that when restricted to each face, $v|_{F}\in \mathbb P_{k+2}(F)$ is an Argyris element~\cite{ArgyrisFriedScharpf1968,BrennerSung2005} and $(\partial_n v)|_F\in\mathbb P_{k+1}(F)$ is a Hermite element~\cite{Ciarlet1978}.

The $\boldsymbol H(\div, \Omega;\mathbb T)$ finite element $\boldsymbol V(K)= \mathbb P_{k-1}(K;\mathbb T)$ is a variant of finite element spaces constructed in~\cite{HuLiang2020} for $k\geq 3$. The space $\mathcal Q(K)=\mathbb P_{k-2}(K;\mathbb R^3)$.

The four local spaces $(W(K),\boldsymbol \Sigma(K), \boldsymbol V(K), \mathcal Q(K))$ will contain polynomial spaces $(\mathbb P_{k+2}, \mathbb P_k, \mathbb P_{k-1}, \mathbb P_{k-2})$ with $2k(k-1)$ non-polynomial shape functions added in $W(K)$ and $\Sigma(K)$ with $k\geq 3$. For the lowest order case, i.e., $k=3$, the dimensions are $(68, 132, 80, 12)$ which are more tractable for implementation.

We show the constructed discrete spaces form a discrete Hessian complex
\begin{equation*}
\mathbb P_1(\Omega)\xrightarrow{\subset} W_h\xrightarrow{\nabla^2}\boldsymbol \Sigma_h\xrightarrow{\curl} \boldsymbol V_h \xrightarrow{\div} \mathcal Q_h\xrightarrow{}\boldsymbol0.
\end{equation*}
Optimal order discretization of the linearized EB system is obtained consequently. 

During the construction, integration by parts is indispensable and therefore the dual complex: div-div complex as well as its polynomial versions are also presented. Finite elements for div-div complex are recently constructed in~\cite{ChenHuang2020,Chen;Huang:2020Finite}.

The rest of this paper is organized as follows. 
Some matrix and vector operations are shown in Section~\ref{sec:matvec}.  In Section~\ref{sec:tensorcomplex} Hessian complex and divdiv complex are presented.
Several polynomial complexes are explored in Section~\ref{sec:tensorpolycomplex}.
A family of conforming virtual element Hessian complexes are constructed in Section~\ref{sec:vemhesscomplex}. In Section~\ref{sec:discreteEB}, the conforming virtual element Hessian complexes are adopted to discretize the linearized EB system.

\section{Matrix and Vector Operations}\label{sec:matvec}
In this section, we shall survey the notation system for operations for vectors and tensors used in the solid mechanic~\cite{Kelly:Mechanics}. In particular, we shall distinguish operators applied to columns and rows of a matrix. The presentation here follows our recent work~\cite{Chen;Huang:2020Finite,Chen;Huang:2021Finite}. 

\subsection{Matrix-vector products}
The matrix-vector product $\boldsymbol A\boldsymbol b$ can be interpreted as the inner product of $\boldsymbol b$ with the row vectors of $\boldsymbol A$. We thus define the dot operator
$\boldsymbol A\cdot \boldsymbol b := \boldsymbol A \boldsymbol b.$ Similarly we can define the row-wise cross product from the right $\boldsymbol A\times \boldsymbol b$. 
Here rigorously speaking when a column vector $\boldsymbol b$ is treated as a row vector, notation $\boldsymbol b^{\intercal}$ should be used. In most places, however, we will sacrifice this precision for the ease of notation. When the vector is on the left of the matrix, the operation is defined column-wise. For example, $\boldsymbol b \cdot \boldsymbol A : = \boldsymbol b^{\intercal}\boldsymbol A$. For dot products, we will still mainly use the conventional notation, e.g. $\boldsymbol b\cdot \boldsymbol A\cdot \boldsymbol c = \boldsymbol b^{\intercal} \boldsymbol A\boldsymbol c$. But for the cross products, we emphasize again the cross product of a vector from the left is column-wise and from the right is row-wise. The transpose rule still works, i.e. $\boldsymbol b\times \boldsymbol A = -(\boldsymbol A^{\intercal}\times \boldsymbol b )^{\intercal}$. Here again, we mix the usage of column vector $\boldsymbol b$ and row vector $\boldsymbol b^{\intercal}$. 

The ordering of performing the row and column products does not matter which leads to the associative rule of the triple products
$$
\boldsymbol b\times \boldsymbol A\times \boldsymbol c := (\boldsymbol b\times \boldsymbol A)\times \boldsymbol c = \boldsymbol b\times (\boldsymbol A\times \boldsymbol c).
$$
Similar rules hold for $\boldsymbol b\cdot \boldsymbol A\cdot \boldsymbol c$ and $\boldsymbol b\cdot \boldsymbol A\times \boldsymbol c$ and thus parentheses can be safely skipped when no differentiation is involved. 

For two column vectors $\boldsymbol u, \boldsymbol v$, the tensor product $\boldsymbol u\otimes \boldsymbol v := \boldsymbol u\boldsymbol v^{\intercal}$ is a matrix which is also known as the dyadic product $\boldsymbol u\boldsymbol v: = \boldsymbol u\boldsymbol v^{\intercal}$ with more clean notation (one $^{\intercal}$ is skipped). The row-wise product and column-wise product with another vector will be applied to the neighboring vector:
\begin{align*}
\boldsymbol x\cdot (\boldsymbol u\boldsymbol v) = (\boldsymbol x\cdot \boldsymbol u) \boldsymbol v^{\intercal}, \quad (\boldsymbol u\boldsymbol v)\cdot \boldsymbol x = \boldsymbol u (\boldsymbol v\cdot \boldsymbol x), \\
\boldsymbol x\times (\boldsymbol u\boldsymbol v) = (\boldsymbol x\times \boldsymbol u) \boldsymbol v, \quad (\boldsymbol u\boldsymbol v)\times \boldsymbol x = \boldsymbol u (\boldsymbol v\times \boldsymbol x).
\end{align*}

\subsection{Differentiation}
We treat Hamilton operator $\nabla = (\partial_1, \partial_2, \partial_3)^{\intercal}$ as a column vector. For a vector function $\boldsymbol u = (u_1, u_2, u_3)^{\intercal}$, $\curl \boldsymbol u= \nabla \times \boldsymbol u$, and $\div \boldsymbol u = \nabla \cdot \boldsymbol u$ are standard differential operations. Define $\nabla \boldsymbol u := \nabla \boldsymbol u^{\intercal} = (\partial_i u_j)$, which can be understood as the dyadic product of Hamilton operator $\nabla$ and column vector $\boldsymbol u$.

Apply these matrix-vector operations to the Hamilton operator $\nabla$, we get column-wise differentiation $\nabla \cdot \boldsymbol A, \nabla \times \boldsymbol A,$
and row-wise differentiation
$\boldsymbol A\cdot \nabla, \boldsymbol A\times \nabla.$ Conventionally, the differentiation is applied to the function after the $\nabla$ symbol. So a more conventional notation is
\begin{align*}
\boldsymbol A\cdot \nabla  : = (\nabla \cdot \boldsymbol A^{\intercal})^{\intercal}, \quad \boldsymbol A\times \nabla : = - (\nabla \times \boldsymbol A^{\intercal})^{\intercal}.
\end{align*}
By moving the differential operator to the right, the notation is simplified and the transpose rule for matrix-vector products can be formally used. Again the right most column vector $\nabla$ is treated as a row vector $\nabla^{\intercal}$ to make the notation cleaner. 

In the literature, differential operators are usually applied row-wisely to tensors. To distinguish with $\nabla$ notation, we define operators in letters which are applied row-wisely
\begin{align*}
\grad \boldsymbol u &:= \boldsymbol u \nabla^{\intercal} = (\partial_j u_i ) = (\nabla \boldsymbol u)^{\intercal},\\
\curl \boldsymbol A &: = - \boldsymbol A\times \nabla = (\nabla \times \boldsymbol A^{\intercal})^{\intercal},\\
\div \boldsymbol A &: = \boldsymbol A\cdot \nabla = (\nabla \cdot \boldsymbol A^{\intercal})^{\intercal}.
\end{align*}

\subsection{Matrix decompositions}
Denote the space of all  $3\times3$ matrices by $\mathbb{M}$, all symmetric $3\times3$ matrices by $\mathbb{S}$, all skew-symmetric $3\times3$ matrices by $\mathbb{K}$, and all trace-free $3\times3$ matrices by $\mathbb{T}$. 
For any matrix $\boldsymbol B\in \mathbb M$, we can decompose it into symmetric and skew-symmetric parts as
$$
\boldsymbol B = {\rm sym}(\boldsymbol B) + {\rm skw}(\boldsymbol B):= \frac{1}{2}(\boldsymbol B + \boldsymbol B^{\intercal}) + \frac{1}{2}(\boldsymbol B - \boldsymbol B^{\intercal}).
$$
We can also decompose it into a direct sum of a trace free matrix and a diagonal matrix as
\begin{equation*}
\boldsymbol B = {\rm dev} \boldsymbol B + \frac{1}{3}\tr(\boldsymbol B)\boldsymbol I := (\boldsymbol B - \frac{1}{3}\tr(\boldsymbol B)\boldsymbol I) + \frac{1}{3}\tr(\boldsymbol B)\boldsymbol I.
\end{equation*}
Define the $\sym\curl$ operator for a matrix $\boldsymbol A$
$$
\sym\curl \boldsymbol A := \frac{1}{2}( \nabla \times \boldsymbol A^{\intercal} + (\nabla \times \boldsymbol A^{\intercal})^{\intercal}) = \frac{1}{2}( \nabla \times \boldsymbol A^{\intercal} - \boldsymbol A\times\nabla ).
$$

We define an isomorphism of $\mathbb R^3$ and the space of skew-symmetric matrices $\mathbb K$ as follows: for a vector $\boldsymbol \omega =
( \omega_1, \omega_2, \omega_3)^{\intercal}
\in \mathbb R^3,$
$$
\mskw \boldsymbol \omega := 
\begin{pmatrix}
 0 & -\omega_3 & \omega_2 \\
\omega_3 & 0 & - \omega_1\\
-\omega_2 & \omega_1 & 0
\end{pmatrix}. 
$$
Obviously $\mskw: \mathbb R^3 \to \mathbb K$ is a bijection. We define $\vskw: \mathbb M\to \mathbb R^3$ by $\vskw := \mskw^{-1}\circ \skw$.

We will use the following identities  which can be verified by direct calculation. 
\begin{align}
\label{eq:skwgrad}
{\rm skw}(\grad\boldsymbol u) &= \frac{1}{2} \mskw (\curl\boldsymbol u),\\
\notag
{\rm skw}(\curl\boldsymbol A) &= \frac{1}{2} \mskw\left[\div(\boldsymbol A^{\intercal})-\grad(\tr(\boldsymbol A))\right], \\
\div \mskw \boldsymbol u &= - \curl \boldsymbol u,\notag\\
\label{eq:divvskw} 2\div\vskw\boldsymbol A &= \tr\curl\boldsymbol A,\\
\notag
\curl (u \boldsymbol I)&=- \mskw \grad(u).
\end{align}
More identities involving the matrix operation and differentiation are summarized in~\cite{ArnoldHu2020}; see also~\cite{Chen;Huang:2020Finite,Chen;Huang:2021Finite}.

\subsection{Projections to a plane}
Given a plane $F$ with normal vector $\boldsymbol n$, for a vector $\boldsymbol v\in \mathbb R^3$, we have the orthogonal decomposition
$$
\boldsymbol v = \Pi_n \boldsymbol v + \Pi_F \boldsymbol v := (\boldsymbol v\cdot \boldsymbol n)\boldsymbol n + (\boldsymbol n\times \boldsymbol v)\times \boldsymbol n.
$$
The matrix representation of $\Pi_n$ is $\boldsymbol n\boldsymbol n^{\intercal}$ and $\Pi_F = I - \boldsymbol n\boldsymbol n^{\intercal}.$ 
The vector $\Pi_F^{\bot}\boldsymbol v :=\boldsymbol n\times \boldsymbol v$ is also on the plane $F$ and is a rotation of $\Pi_F \boldsymbol v$ by $90^{\circ}$ counter-clockwise with respect to $\boldsymbol n$. 
We treat Hamilton operator $\nabla = (\partial_1, \partial_2, \partial_3)^{\intercal}$ as a column vector and define
$$
\nabla_F^{\bot} := \boldsymbol n\times \nabla, \quad \nabla_F: = \Pi_F \nabla = -\boldsymbol n\times (\boldsymbol n\times \nabla).
$$
For a scalar function $v$,
\begin{align*}
\grad_F v : = \nabla_F v = \Pi_F (\nabla v), \\
 \curl_F v := \nabla_F^{\bot} v = \boldsymbol n \times \nabla v,
\end{align*}
 are the surface gradient of $v$ and surface $\curl$, respectively. For a vector function $\boldsymbol v$, $\nabla_F\cdot \boldsymbol v$ is the surface divergence
$$
\div_F\boldsymbol v := \nabla_F\cdot \boldsymbol v = \nabla_F\cdot(\Pi_F\boldsymbol v).
$$
By the cyclic invariance of the mix product and the fact $\boldsymbol  n$ is constant, the surface rot operator is
\begin{equation*}
{\rm rot}_F \boldsymbol  v := \nabla_F^{\bot}\cdot \boldsymbol  v = (\boldsymbol  n\times \nabla)\cdot \boldsymbol  v = \boldsymbol  n\cdot (\nabla \times \boldsymbol  v),
\end{equation*}
which is the normal component of $\nabla \times \boldsymbol  v$. 
The tangential trace of $\nabla \times \boldsymbol  v$ is 
\begin{equation*}
\boldsymbol  n\times (\nabla \times \boldsymbol  v) = \nabla (\boldsymbol  n\cdot \boldsymbol  v) - \partial_n \boldsymbol  v. 
\end{equation*}
By definition,
\begin{equation*}
{\rm rot}_F \boldsymbol  v = - \div_F (\boldsymbol  n\times \boldsymbol  v), \quad
\div_F \boldsymbol  v = {\rm rot}_F (\boldsymbol  n\times \boldsymbol  v).
\end{equation*}
Note that the three dimensional $\curl$ operator restricted to a two dimensional plane $F$ results in two operators: $\curl_F$ maps a scalar to a vector, which is a rotation of $\grad_F$, and $\rot_F$ maps a vector to a scalar which can be thought as a rotated version of $\div_F$. The surface differentiations satisfy the property $\div_F\curl_F = 0$ and $\rot_F\grad_F = 0$ and when $F$ is simply connected, $\ker(\div_F) = {\rm img}(\curl_F)$ and $\ker(\rot_F) = {\rm img}(\grad_F)$.

Differentiation for two dimensional tensors on face $F$ can be defined similarly. 

\section{Two Hilbert Complexes for Tensors}\label{sec:tensorcomplex}
In this section we shall present two Hilbert complexes for tensors: the Hessian complex and the divdiv complex. They are dual to each other. The Hessian complex will be used for the construction of shape function spaces and the divdiv complex for the degrees of freedom. 

Recall that a Hilbert complex is a sequence of Hilbert spaces $\{\mathcal V_i \}$ connected by a sequence of closed densely defined linear operators $\{\dd_i\}$ 
$$
0 \stackrel{}{\longrightarrow} \mathcal V_1 \stackrel{\dd_1}{\longrightarrow} \mathcal V_2 \stackrel{\dd_2}{\longrightarrow} \cdots \stackrel{\dd_{n-2}}{\longrightarrow}\mathcal V_{n-1}\stackrel{\dd_{n-1}}{\longrightarrow} \mathcal V_n \longrightarrow 0,
$$
satisfying the property $\img (\dd_i)\subseteq \ker(\dd_{i+1})$, i.e., $\dd_{i+1}\circ \dd_i = 0$. In this paper, we shall consider domain complexes only, i.e., ${\rm dom}(\dd_i) = \mathcal V_i$. The complex is called an exact sequence if $\img (\dd_i)=  \ker(\dd_{i+1})$ for $i=1, \ldots, n$. We usually skip the first $0$ in the complex and use the embedding operator to indicate $\dd_1$ is injective. We refer to~\cite{Arnold:2018Finite} for background on Hilbert complexes. 

\subsection{Hessian complexes}
The Hessian complex in three dimensions reads as~\cite{ArnoldHu2020,PaulyZulehner2020}
\begin{equation}\label{eq:hesscomplex}
\mathbb P_1(\Omega)\xrightarrow{\subset} H^{2}(\Omega)\xrightarrow{\hess}\boldsymbol H(\curl, \Omega;\mathbb S)\xrightarrow{\curl} \boldsymbol H(\div, \Omega;\mathbb T) \xrightarrow{\div} \boldsymbol L^2(\Omega;\mathbb R^3)\xrightarrow{}\boldsymbol0.
\end{equation}

For the completeness we shall prove the exactness following~\cite{PaulyZulehner2020} and refer to~\cite{ArnoldHu2020} for a systematical way of deriving complexes from complexes.  

\begin{lemma}
Assume $\Omega$ is a bounded Lipschitz domain in $\mathbb R^3$.
It holds
\begin{equation}\label{eq:divH1TontoL2}
\div\boldsymbol H^1(\Omega;\mathbb T)=\boldsymbol L^2(\Omega; \mathbb R^3).
\end{equation}
\end{lemma}
\begin{proof}
First consider $\boldsymbol v=\nabla w\in\boldsymbol L^2(\Omega; \mathbb R^3)$ with $w\in H^1(\Omega)$. There exists $\boldsymbol\phi\in \boldsymbol H^2(\Omega; \mathbb R^3)$ statisfying $2\div\boldsymbol\phi=-3w$. Take $\boldsymbol\tau=w\boldsymbol I+\curl\mskw\boldsymbol\phi\in\boldsymbol H^1(\Omega;\mathbb M)$. It is obvious that $\div\boldsymbol\tau=\div(w\boldsymbol I)=\boldsymbol v$. It follows from~\eqref{eq:divvskw} that
\[
\tr\boldsymbol\tau=3w+\tr\curl\mskw\boldsymbol\phi=3w+2\div\vskw\mskw\boldsymbol\phi=3w+2\div\boldsymbol\phi=0.
\]

Next consider general $\boldsymbol v\in\boldsymbol L^2(\Omega; \mathbb R^3)$. There exists $\boldsymbol\tau_1\in\boldsymbol H^1(\Omega; \mathbb M)$ satisfying $\div\boldsymbol\tau_1=\boldsymbol v$. Then there exists $\boldsymbol\tau_2\in\boldsymbol H^1(\Omega; \mathbb T)$ satisfying $\div\boldsymbol\tau_2=\frac{1}{3}\nabla(\tr\boldsymbol\tau_1)$. Now take $\boldsymbol\tau=\dev\boldsymbol\tau_1+\boldsymbol\tau_2\in\boldsymbol H^1(\Omega; \mathbb T)$. We have
\[
\div\boldsymbol\tau=\div(\dev\boldsymbol\tau_1)+\div\boldsymbol\tau_2=\div(\dev\boldsymbol\tau_1)+\frac{1}{3}\nabla(\tr\boldsymbol\tau_1)=\div\boldsymbol\tau_1=\boldsymbol v.
\]
Thus~\eqref{eq:divH1TontoL2} follows.
\end{proof}

\begin{lemma}
Assume $\Omega$ is a bounded and topologically trivial Lipschitz domain in $\mathbb R^3$.
It holds
\begin{equation}\label{eq:curlH1SontoDivfree}
\curl\boldsymbol H^1(\Omega;\mathbb S)=\boldsymbol H(\div, \Omega;\mathbb T)\cap\ker(\div).
\end{equation}
\end{lemma}
\begin{proof}
By \cite[Theorem 1.1]{CostabelMcIntosh2010}, for any $\boldsymbol\tau\in\boldsymbol H(\div, \Omega;\mathbb T)\cap\ker(\div)$, there exists $\boldsymbol \sigma_1 \in\boldsymbol H^1(\Omega;\mathbb M)$  such that
\[
\boldsymbol\tau=\curl\boldsymbol\sigma_1.
\]
Thanks to~\eqref{eq:divvskw}, we have
\[
2\div\vskw\boldsymbol\sigma_1=\tr\curl\boldsymbol\sigma_1=\tr\boldsymbol\tau=0.
\]
Hence there exsits $\boldsymbol v\in \boldsymbol H^2(\Omega; \mathbb R^3)$ such that $\vskw\boldsymbol\sigma_1=\frac{1}{2}\curl\boldsymbol v$. Then apply $\mskw$ and use~\eqref{eq:skwgrad} to get 
\[
\skw\boldsymbol\sigma_1=\frac{1}{2}\mskw\curl\boldsymbol v=\skw(\grad\boldsymbol v).
\]
Taking $\boldsymbol\sigma=\boldsymbol\sigma_1-\grad\boldsymbol v$, we have $\boldsymbol\sigma\in \boldsymbol H^1(\Omega; \mathbb S)$ and $\curl\boldsymbol\sigma=\boldsymbol\tau$.
\end{proof}

\begin{theorem}\label{thm:hessiancomplex}
Assume $\Omega$  is a bounded and topologically trivial Lipschitz domain in $\mathbb R^3$. Then~\eqref{eq:hesscomplex} is a Hilbert complex and  exact sequence.   
\end{theorem}
\begin{proof}
It is obvious that~\eqref{eq:hesscomplex} is a complex and $H^{2}(\Omega)\cap\ker(\hess)=\mathbb P_1(\Omega)$.
As results of~\eqref{eq:divH1TontoL2} and~\eqref{eq:curlH1SontoDivfree}, we have
\[
\div\boldsymbol H(\div, \Omega;\mathbb T)=\boldsymbol L^2(\Omega; \mathbb R^3),\quad \curl\boldsymbol H(\curl, \Omega;\mathbb S)=\boldsymbol H(\div, \Omega;\mathbb T)\cap\ker(\div).
\]

We only need to prove $\boldsymbol H(\curl, \Omega;\mathbb S)\cap\ker(\curl)=\hess\,H^2(\Omega)$.
For any $\boldsymbol\sigma\in\boldsymbol H(\curl, \Omega;\mathbb S)\cap\ker(\curl)$, there exists $\boldsymbol v \in\boldsymbol H^1(\Omega;\mathbb R^3)$ such that
\[
\boldsymbol\sigma=\grad\boldsymbol v.
\]
Since $\boldsymbol\sigma$ is symmetric, by~\eqref{eq:skwgrad}, we have
\[
\mskw(\curl\boldsymbol v)=2\skw(\grad\boldsymbol v)=2\skw(\boldsymbol\sigma)=\boldsymbol0,
\]
which means $\curl\boldsymbol v=\boldsymbol0$. Hence there exists $w\in H^2(\Omega)$ that $\boldsymbol v=\nabla w$ and consequently $\boldsymbol\sigma=\hess\,w\in\hess\,H^2(\Omega)$.
\end{proof}

As a result of the Hessian complex~\eqref{eq:hesscomplex}, we have the Poincar\'e inequality~\cite[the inequality above (14)]{ArnoldHu2020}
\begin{equation}\label{eq:Poincareieqlity}
\|\boldsymbol\tau\|_0\lesssim \|\curl \, \boldsymbol\tau\|_0
\end{equation}
for any $\boldsymbol\tau\in\boldsymbol H(\curl, \Omega;\mathbb S)$ satisfying
\[
(\boldsymbol\tau,\nabla^2w)=0\quad\forall~w\in H^2(\Omega).
\]

When $\Omega\subset\mathbb R^2$,
the Hessian complex in two dimensions becomes
\begin{equation*}
\mathbb P_1(\Omega)\xrightarrow{\subset} H^{2}(\Omega)\xrightarrow{\hess}\boldsymbol H(\rot, \Omega;\mathbb S)\xrightarrow{\rot}  \boldsymbol L^2(\Omega;\mathbb R^2)\xrightarrow{}\boldsymbol0,
\end{equation*}
which is a rotation of the elasticity complex~\cite{Eastwood2000,ArnoldWinther2002}.

\subsection{divdiv complexes}
The $\div\div$ complex in three dimensions reads as~\cite{ArnoldHu2020,PaulyZulehner2020}
\begin{equation}\label{eq:divdivcomplex3d}
\resizebox{.915\hsize}{!}{$
\boldsymbol{RT}\xrightarrow{\subset} \boldsymbol H^1(\Omega;\mathbb R^3)\xrightarrow{\dev\grad}\boldsymbol H(\sym\curl,\Omega;\mathbb T)\xrightarrow{\sym\curl} \boldsymbol H(\div\div, \Omega;\mathbb S) \xrightarrow{\div\div} L^2(\Omega)\xrightarrow{}0,
$}
\end{equation}
where $\boldsymbol{RT}:=\{a\boldsymbol x + \boldsymbol b: a\in \mathbb R, \boldsymbol b \in \mathbb R^3\}$ is the lowest order Raviart-Thomas space.

A proof of the following theorem can be found in~\cite{ArnoldHu2020,PaulyZulehner2020,Chen;Huang:2020Finite}. 
\begin{theorem}\label{thm:divdivcomplex}
Assume $\Omega$  is a bounded and topologically trivial Lipschitz domain in $\mathbb R^3$. Then~\eqref{eq:divdivcomplex3d} is a Hilbert complex and  exact sequence.   
\end{theorem}

When $\Omega\subset\mathbb R^2$,
the $\div\div$ complex in two dimensions becomes (cf.~\cite{ChenHuang2018})
\begin{equation*}
\boldsymbol {RT}\xrightarrow{\subset} \boldsymbol H^1(\Omega;\mathbb R^2)\xrightarrow{\sym\curl} \boldsymbol{H}(\div\div,\Omega; \mathbb{S}) \xrightarrow{\div {\div}} L^2(\Omega)\xrightarrow{}0.
\end{equation*}

\section{Polynomial Complexes for Tensors}\label{sec:tensorpolycomplex}
In this section we consider Hessian and divdiv polynomial complexes on a bounded and topologically trivial domain $D\subset \mathbb R^3$.
Without loss of generality, we assume $(0,0,0) \in D$.

Given a non-negative integer $k$, 
let $\mathbb P_k(D)$ stand for the set of all polynomials in $D$ with the total degree no more than $k$, and $\mathbb P_k(D; \mathbb{X})$ denote the tensor or vector version.
Let $\mathbb H_k(D):=\mathbb P_k(D)/\mathbb P_{k-1}(D)$ be the space of functions spanned by the homogenous polynomials of degree $k$. 
Denote by $Q_k^D$ the $L^2$-orthogonal projector onto $\mathbb P_k(D)$, and $\boldsymbol Q_k^D$ the tensor or vector version.

\subsection{De Rham and Koszul polynomial complexes}
First we recall the polynomial de Rham complex
\begin{equation}\label{eq:polydeRham}
\mathbb R \xrightarrow{\subset} \mathbb P_{k+1}(D) 
\xrightarrow{\nabla}
\mathbb P_{k}(D;\mathbb R^3)
\xrightarrow{\nabla \times}
\mathbb P_{k-1}(D;\mathbb R^3)
\xrightarrow{\nabla\cdot}
\mathbb P_{k-2}(D) \to 0,
\end{equation}
 and the Koszul complex going backwards
\begin{equation}\label{eq:Koszul}
\mathbb P_{k+1}(D) 
\xleftarrow{\boldsymbol x\cdot}
\mathbb P_{k}(D;\mathbb R^3)
\xleftarrow{\boldsymbol x\times}
\mathbb P_{k-1}(D;\mathbb R^3)
\xleftarrow{\boldsymbol x} \mathbb P_{k-2}(D) \xleftarrow{} 0.
\end{equation}
Those two complexes can be combined into one
\begin{equation}\label{eq:deRhamcomplex3dPolydouble}
\resizebox{.91\hsize}{!}{$
\xymatrix{
\mathbb R\ar@<0.4ex>[r]^-{\subset} & \mathbb P_{k+1}(D)\ar@<0.4ex>[r]^-{\nabla} & \mathbb P_{k}(D;\mathbb R^3)\ar@<0.4ex>[r]^-{\nabla \times}\ar@<0.4ex>[l]^-{\boldsymbol x\cdot}  & \mathbb P_{k-1}(D;\mathbb R^3) \ar@<0.4ex>[r]^-{\nabla \cdot}\ar@<0.4ex>[l]^-{\boldsymbol x\times} & \mathbb P_{k-2}(D)  \ar@<0.4ex>[l]^-{\boldsymbol x}
& 0 \ar@<0.4ex>[l]^-{\supset} }.
$}
\end{equation}
We refer to~\cite{Arnold;Falk;Winther:2006Finite} for a systematical derivation of~\eqref{eq:polydeRham}-\eqref{eq:Koszul} and focus on two decompositions of vector polynomial spaces $\mathbb P_{k}(D;\mathbb R^3)$ based on~\eqref{eq:deRhamcomplex3dPolydouble}. One subspace is the range space of a differential operator in the de Rham complex from left to right and another is the range space of the Koszul operator.

The first one is, for an integer $k\geq 1$,
\begin{equation*}
\mathbb P_{k}(D;\mathbb R^3) = \nabla \mathbb P_{k+1}(D)\oplus \boldsymbol x\times \mathbb P_{k-1}(D;\mathbb R^3), 
\end{equation*}
which leads to 
$$
\mathbb P_{k}(D;\mathbb R^3)   = \nabla \mathbb H_{k+1}(D) \oplus \mathcal{ND}_{k-1},
$$
where 
\begin{equation*}
 \mathcal{ND}_{k-1} := \mathbb P_{k-1}(D;\mathbb R^3) \oplus \boldsymbol x\times \mathbb H_{k-1}(D;\mathbb R^3) = \mathbb P_{k-1}(D;\mathbb R^3) + \boldsymbol x\times \mathbb P_{k-1}(D;\mathbb R^3)
\end{equation*}
is the first family of N\'ed\'elec element~\cite{Nedelec:1980finite}. 
Note that the component $\boldsymbol x\times \mathbb H_{k-1}(D;\mathbb R^3)$ can be also written as $\ker(\boldsymbol x\cdot)\cap \mathbb H_k(D;\mathbb R^3)$ by the exactness of the Koszul complex~\eqref{eq:Koszul}, which unifies the notation in both two and three dimensions. 

The second decomposition is, for an integer $k\geq 1$,
\begin{equation}\label{eq:RTdecomposition}
\mathbb P_{k}(D;\mathbb R^3) = \nabla \times \mathbb P_{k+1}(D;\mathbb R^3)\oplus \boldsymbol x \mathbb P_{k-1}(D),
\end{equation}
which leads to 
$$
\mathbb P_{k}(D;\mathbb R^3) 
= \nabla \times \mathbb H_{k+1}(D;\mathbb R^3) \oplus \mathcal{RT}_{k-1},
$$
where
\begin{equation*}
\mathcal{RT}_{k-1} : = \mathbb P_{k-1}(D;\mathbb R^3)\oplus \boldsymbol x \mathbb H_{k-1}(D) = \mathbb P_{k-1}(D;\mathbb R^3)+ \boldsymbol x \mathbb P_{k-1}(D)
\end{equation*}
is the Raviart-Thomas face element in three dimensions~\cite{RaviartThomas1977,Nedelec:1986family}.

\subsection{Hessian polynomial complexes}

By the Euler's formula, for an integer $k\geq 0$,
\begin{equation}\label{eq:homogeneouspolyprop}
\boldsymbol x\cdot\nabla q=kq\quad\forall~q\in\mathbb H_k(D).
\end{equation}
Due to~\eqref{eq:homogeneouspolyprop}, for any $q\in\mathbb P_k(D)$ satisfying $\boldsymbol x\cdot\nabla q+q=0$, we have $q=0$.
And
\begin{align}\label{eq:radialderivativeprop}
\mathbb P_k(D)\cap\ker(\boldsymbol x\cdot\nabla) & =\mathbb P_0(D),\\
\label{eq:radialderivativeprop1}
\mathbb P_k(D)\cap\ker(\boldsymbol x\cdot\nabla +\ell) &=0
\end{align}
for any positive number $\ell$.

\begin{lemma}\label{lem:divtracetensorsurjective}
The operator $\div: \dev(\mathbb P_{k}(D;\mathbb R^3)\boldsymbol x^{\intercal})\to \mathbb P_{k}(D;\mathbb R^3)$ is bijective.
\end{lemma}
\begin{proof}
Since $\div\dev(\mathbb P_{k}(D;\mathbb R^3)\boldsymbol x^{\intercal})\subseteq\mathbb P_{k}(D;\mathbb R^3)$ and 
$$\dim\dev(\mathbb P_{k}(D;\mathbb R^3)\boldsymbol x^{\intercal})=\dim\mathbb P_{k}(D;\mathbb R^3),$$
it sufficies to show that $\div: \dev(\mathbb P_{k}(D;\mathbb R^3)\boldsymbol x^{\intercal})\to \mathbb P_{k}(D;\mathbb R^3)$ is injective.

For any $\boldsymbol q\in\mathbb P_{k}(D;\mathbb R^3)$ satisfying $\div\dev(\boldsymbol q\boldsymbol x^{\intercal})=0$, we have
\begin{equation}\label{eq:20200514}
\div(\boldsymbol q\boldsymbol x^{\intercal})-\frac{1}{3}\nabla(\boldsymbol x^{\intercal}\boldsymbol q)=\div(\dev(\boldsymbol q\boldsymbol x^{\intercal}))=\boldsymbol 0.
\end{equation}
Since $\boldsymbol x^{\intercal}\div(\boldsymbol q\boldsymbol x^{\intercal})=(\boldsymbol x\cdot\nabla)(\boldsymbol x^{\intercal}\boldsymbol q)+2\boldsymbol x^{\intercal}\boldsymbol q$, we obtain
\[
\left(\boldsymbol x\cdot\nabla+3\right)(\boldsymbol x^{\intercal}\boldsymbol q)=0.
\]
By~\eqref{eq:radialderivativeprop1}, we have $\boldsymbol x^{\intercal}\boldsymbol q=0$. In turn, it follows from~\eqref{eq:20200514} that $(\boldsymbol x\cdot\nabla+3)\boldsymbol q=\div(\boldsymbol q\boldsymbol x^{\intercal})=\boldsymbol 0$, which together with~\eqref{eq:radialderivativeprop1} gives $\boldsymbol q=\boldsymbol 0$.
\end{proof}

\begin{lemma}
For $k\in \mathbb N, k\geq 2$, the polynomial Hessian complex
\begin{equation}\label{eq:hesscomplex3dPoly}
\resizebox{.91\hsize}{!}{$
\mathbb P_1(D)\xrightarrow{\subset} \mathbb P_{k+2}(D)\xrightarrow{\hess}\mathbb P_{k}(D;\mathbb S)\xrightarrow{\curl} \mathbb P_{k-1}(D;\mathbb T) \xrightarrow{\div} \mathbb P_{k-2}(D;\mathbb R^3)\xrightarrow{}\boldsymbol0
$}
\end{equation}
is exact.
\end{lemma}
\begin{proof}
It is obvious $\nabla^2(\mathbb P_{k+2}(D))\subseteq \mathbb P_{k}(D;\mathbb S)\cap \ker(\curl)$. By identity~\eqref{eq:divvskw},
$$\tr(\curl\boldsymbol\tau)=2\div(\vskw\boldsymbol\tau)\quad\forall~\boldsymbol\tau\in\boldsymbol H^1(D;\mathbb M).$$
Hence we have $\curl (\mathbb P_{k}(D;\mathbb S))\subseteq \mathbb P_{k-1}(D;\mathbb T) \cap \ker(\div)$. Therefore~\eqref{eq:hesscomplex3dPoly} is a complex.

We then verify this complex is exact. By the polynomial version of de Rham complex~\eqref{eq:polydeRham}, we have $\hess\,\mathbb P_{k+2}(D)=\mathbb P_{k}(D;\mathbb S)\cap\ker(\curl)$, and
\[
\dim\curl\mathbb P_{k}(D;\mathbb S)=\dim\mathbb P_{k}(D;\mathbb S)-\dim\hess\,\mathbb P_{k+2}(D)=\frac{1}{6}k(k+1)(5k+19).
\]
Thanks to Lemma~\ref{lem:divtracetensorsurjective}, we get $\div\mathbb P_{k-1}(D;\mathbb T)=\mathbb P_{k-2}(D;\mathbb R^3)$. And then
\[
\resizebox{.99\hsize}{!}{$
\dim(\mathbb P_{k-1}(D;\mathbb T)\cap\ker(\div))=\dim\mathbb P_{k-1}(D;\mathbb T)-\dim\mathbb P_{k-2}(D;\mathbb R^3)=\dim\curl\mathbb P_{k}(D;\mathbb S),
$}
\]
which means $\mathbb P_{k-1}(D;\mathbb T)\cap\ker(\div)=\curl\mathbb P_{k}(D;\mathbb S)$.
Therefore the complex~\eqref{eq:hesscomplex3dPoly} is exact.
\end{proof}

Define operator $\pi_{1}: \mathcal C^1(D)\to\mathbb P_{1}(D)$ as
\[
\pi_{1}v:=v(0,0,0)+\boldsymbol x^{\intercal}(\nabla v)(0,0,0).
\]
It is exactly the first order Taylor polynomial of $v$ at $(0,0,0)$. 
Obviously
\begin{equation}\label{eq:pi1prop}
\pi_{1}v=v\quad \forall~v\in\mathbb P_{1}(D).
\end{equation}

We present the following Koszul-type complex associated to the Hessian complex.
\begin{lemma}
For $k\in \mathbb N, k\geq 2$, the polynomial complex
\begin{equation}\label{eq:hessKoszulcomplex3dPoly}
\resizebox{.91\hsize}{!}{$
0\xrightarrow{\subset}\mathbb P_{k-2}(D;\mathbb R^3) \xrightarrow{\dev(\boldsymbol v\boldsymbol x^{\intercal})} \mathbb P_{k-1}(D;\mathbb T) \xrightarrow{\sym(\boldsymbol\tau\times\boldsymbol x)} \mathbb P_{k}(D;\mathbb S)\xrightarrow{\boldsymbol x^{\intercal}\boldsymbol\tau\boldsymbol x} \mathbb P_{k+2}(D)\xrightarrow{\pi_1}\mathbb P_{1}(D)
$}
\end{equation}
is exact.
\end{lemma}
\begin{proof}
For any $\boldsymbol v\in\mathbb P_{k-2}(D;\mathbb R^3)$, it follows
\[
\sym((\dev(\boldsymbol v\boldsymbol x^{\intercal}))\times\boldsymbol x)=\sym((\boldsymbol v\boldsymbol x^{\intercal})\times\boldsymbol x)-\frac{1}{3}(\boldsymbol x^{\intercal}\boldsymbol v)\sym(\boldsymbol I\times\boldsymbol x)=\boldsymbol0.
\]
For any $\boldsymbol \tau\in\mathbb P_{k-1}(D;\mathbb T)$, we have
\[
\boldsymbol x^{\intercal}(\sym(\boldsymbol\tau\times\boldsymbol x))\boldsymbol x=\boldsymbol x^{\intercal}(\boldsymbol\tau\times\boldsymbol x)\boldsymbol x=0.
\]
It is trivial that $\pi_1(\boldsymbol x^{\intercal}\boldsymbol\tau\boldsymbol x)=0$ for any $\boldsymbol\tau\in\mathbb P_{k}(D;\mathbb S)$. Thus~\eqref{eq:hessKoszulcomplex3dPoly} is a complex.

Next we prove that the complex~\eqref{eq:hessKoszulcomplex3dPoly} is exact.
By the Taylor's theorem, we get $\mathbb P_{k+2}(D)\cap\ker(\pi_1)=\boldsymbol x^{\intercal}\mathbb P_{k}(D;\mathbb S)\boldsymbol x$, and
\[
\dim\boldsymbol x^{\intercal}\mathbb P_{k}(D;\mathbb S)\boldsymbol x=\dim\mathbb P_{k+2}(D)-4=\frac{1}{6}(k+5)(k+4)(k+3)-4.
\]

For any $\boldsymbol\tau\in\mathbb P_{k}(D;\mathbb S)$ satisfying $\boldsymbol x^{\intercal}\boldsymbol\tau\boldsymbol x=0$, there exists $\boldsymbol q\in\mathbb P_{k-1}(D;\mathbb R^3)$ such that $\boldsymbol\tau\boldsymbol x=\boldsymbol q\times \boldsymbol x=(\mskw\boldsymbol q) \boldsymbol x$, that is $(\boldsymbol\tau-\mskw\boldsymbol q)\boldsymbol x=\boldsymbol 0$. As a result, there exists $\boldsymbol\varsigma\in\mathbb P_{k}(D;\mathbb M)$ such that
\[
\boldsymbol\tau=\mskw\boldsymbol q+\boldsymbol\varsigma\times\boldsymbol x.
\]
From the symmetry of $\boldsymbol\tau$, we obtain
\[
\boldsymbol\tau=\sym(\mskw\boldsymbol q+\boldsymbol\varsigma\times\boldsymbol x)=\sym(\boldsymbol\varsigma\times\boldsymbol x)=\sym(\dev\boldsymbol\varsigma\times\boldsymbol x)\in\sym(\mathbb P_{k-1}(D;\mathbb T)\times\boldsymbol x).
\]
Hence
\[
\dim\sym(\mathbb P_{k-1}(D;\mathbb T)\times\boldsymbol x)=\mathbb P_{k}(D;\mathbb S)-\dim\boldsymbol x^{\intercal}\mathbb P_{k}(D;\mathbb S)\boldsymbol x=\frac{1}{6}k(k+1)(5k+19).
\]

Since $\dim\dev(\mathbb P_{k-2}(D;\mathbb R^3)\boldsymbol x^{\intercal})=\dim\mathbb P_{k-2}(D;\mathbb R^3)$,
we have
\[
\dim\mathbb P_{k-1}(D;\mathbb T)=\dim\dev(\mathbb P_{k-2}(D;\mathbb R^3)\boldsymbol x^{\intercal})+\dim\sym(\mathbb P_{k-1}(D;\mathbb T)\times\boldsymbol x).
\]
Thus the complex~\eqref{eq:hessKoszulcomplex3dPoly} is exact.
\end{proof}

Combining the two complexes~\eqref{eq:hesscomplex3dPoly} and~\eqref{eq:hessKoszulcomplex3dPoly} yields
\begin{equation*}
\resizebox{.91\hsize}{!}{$
\xymatrix{
\mathbb P_1(D)\ar@<0.4ex>[r]^-{\subset} & \mathbb P_{k+2}(D)\ar@<0.4ex>[r]^-{\hess}\ar@<0.4ex>[l]^-{\pi_{1}} & \mathbb P_{k}(D;\mathbb S)\ar@<0.4ex>[r]^-{\curl}\ar@<0.4ex>[l]^-{\boldsymbol x^{\intercal}\boldsymbol\tau\boldsymbol x}  & \mathbb P_{k-1}(D;\mathbb T) \ar@<0.4ex>[r]^-{\div}\ar@<0.4ex>[l]^-{\sym(\boldsymbol\tau\times\boldsymbol x)} & \mathbb P_{k-2}(D;\mathbb R^3)  \ar@<0.4ex>[r]^-{} \ar@<0.4ex>[l]^-{\dev(\boldsymbol v\boldsymbol x^{\intercal})}
& \boldsymbol0 \ar@<0.4ex>[l]^-{\supset} }.
$}
\end{equation*}
Unlike the Koszul complex for vectors functions, we do not have the identity property applied to homogenous polynomials. Fortunately decomposition of polynomial spaces using Koszul and differential operators still holds.

It follows from~\eqref{eq:pi1prop} and the complex~\eqref{eq:hessKoszulcomplex3dPoly} that
\begin{equation*}
\mathbb P_{k+2}(D)=\boldsymbol x^{\intercal}\mathbb P_{k}(D;\mathbb S)\boldsymbol x\oplus\mathbb P_1(D), \quad k\geq 0.
\end{equation*}
Then we give the following decompositions for the polynomial tensor spaces $\mathbb P_{k}(D;\mathbb S)$ and $\mathbb P_{k-1}(D;\mathbb T)$. Again one subspace is the range space of a differential operator in the Hessian complex from left-to-right and another is the range space in the Koszul type complex from right-to-left.

\begin{lemma}
For $k\in \mathbb N$, we have the decompositions
\begin{align}\label{eq:hesspolyspacedecomp2}
\mathbb P_{k}(D;\mathbb S) &=\hess\, \mathbb P_{k+2}(D)\oplus\sym(\mathbb P_{k-1}(D;\mathbb T)\times\boldsymbol x) & k\geq 1,\\
\label{eq:hesspolyspacedecomp3}
\mathbb P_{k-1}(D;\mathbb T) &=\curl \, \mathbb P_{k}(D;\mathbb S)\oplus\dev(\mathbb P_{k-2}(D;\mathbb R^3)\boldsymbol x^{\intercal}) & k\geq 2.
\end{align}
\end{lemma}
\begin{proof}
Noting that the dimension of space in the left hand side is the summation of the dimension of two subspaces in the right hand side in~\eqref{eq:hesspolyspacedecomp2} and~\eqref{eq:hesspolyspacedecomp3}, we only need to prove the sum is direct. The direct sum of~\eqref{eq:hesspolyspacedecomp3} follows from Lemma~\ref{lem:divtracetensorsurjective}. We then focus on \eqref{eq:hesspolyspacedecomp2}.

For any $\boldsymbol\tau=\nabla^2q$ with $q\in\mathbb P_{k+2}(D)$ satisfying $\boldsymbol\tau\in\sym(\mathbb P_{k-1}(D;\mathbb T)\times\boldsymbol x)$, it follows from the fact $(\boldsymbol x\cdot\nabla)\boldsymbol x=\boldsymbol x$ that
\[
(\boldsymbol x\cdot\nabla)(\boldsymbol x\cdot\nabla q-q)=(\boldsymbol x\cdot\nabla)(\boldsymbol x\cdot\nabla q)-\boldsymbol x\cdot\nabla q=\boldsymbol x^{\intercal}((\boldsymbol x\cdot\nabla)\nabla q)=\boldsymbol x^{\intercal}(\nabla^2q)\boldsymbol x=0.
\]
Applying~\eqref{eq:radialderivativeprop} to get $\boldsymbol x\cdot\nabla q-q\in\mathbb P_0(K)$, which together with~\eqref{eq:homogeneouspolyprop} gives $q\in\mathbb P_{1}(D)$. Thus the decomposition~\eqref{eq:hesspolyspacedecomp2} holds.
\end{proof}

When $D\subset\mathbb R^2$,
the Hessian polynomial complex in two dimensions
\begin{equation}\label{eq:hesscomplex2dPoly}
\mathbb P_{1}(D)\xrightarrow{\subset} \mathbb P_{k+2}(D)\xrightarrow{\hess} \mathbb P_{k}(D;\mathbb S) \xrightarrow{{\rot}} \mathbb P_{k-1}(D;\mathbb R^2)\xrightarrow{}0
\end{equation}
has been proved in~\cite{ChenHuang2020}, which is a rotation of the elasticity polynomial complex~\cite{ArnoldWinther2002}.

\subsection{Divdiv Polynomial complexes}
In this subsection we present divdiv polynomial complexes derived in~\cite{ChenHuang2020,Chen;Huang:2020Finite} and refer to \cite{Chen;Huang:2020Finite} for proofs.

\begin{lemma}
For $k\in \mathbb N, k\geq 2$, the polynomial complex
\begin{equation}\label{eq:divdivcomplex3dPoly}
\resizebox{.915\hsize}{!}{$
\boldsymbol{RT}\xrightarrow{\subset} \mathbb P_{k+2}(D; \mathbb R^3)\xrightarrow{\dev\grad}\mathbb P_{k+1}(D; \mathbb T)\xrightarrow{\sym\curl} \mathbb P_k(D; \mathbb S) \xrightarrow{\div\div} \mathbb P_{k-2}(D)\xrightarrow{}0
$}
\end{equation}
is exact.
\end{lemma}

Define operator $\boldsymbol \pi_{RT}: \mathcal C^1(D; \mathbb R^3)\to \boldsymbol{RT}$ as
\[
\boldsymbol \pi_{RT}\boldsymbol  v:=\boldsymbol  v(0,0,0)+\frac{1}{3}(\div\boldsymbol  v)(0,0,0)\boldsymbol  x.
\]
Apparently
\begin{equation}\label{eq:piRTprop}
\boldsymbol \pi_{RT}\boldsymbol  v=\boldsymbol  v\quad \forall~\boldsymbol  v\in\boldsymbol{RT}.
\end{equation}
We have the following Koszul-type complex. 
\begin{lemma}\label{lem:Koszul}
For $k\in \mathbb N, k\geq 2$, the polynomial complex
\begin{equation}\label{eq:divdivKoszulcomplex3dPoly}
\resizebox{.915\hsize}{!}{$
0\xrightarrow{\subset}\mathbb P_{k-2}(D) \xrightarrow{\boldsymbol x\boldsymbol x^{\intercal}} \mathbb P_k(D; \mathbb S) \xrightarrow{\times\boldsymbol x} \mathbb P_{k+1}(D; \mathbb T)\xrightarrow{\cdot\boldsymbol x} \mathbb P_{k+2}(D; \mathbb R^3)\xrightarrow{\boldsymbol \pi_{RT}}\boldsymbol{RT}\xrightarrow{}\boldsymbol0
$}
\end{equation}
is exact.
\end{lemma}

Those two complexes~\eqref{eq:divdivcomplex3dPoly} and~\eqref{eq:divdivKoszulcomplex3dPoly} are connected as
\begin{equation*}
\resizebox{.99\hsize}{!}{$
\xymatrix{
\boldsymbol{RT}\ar@<0.4ex>[r]^-{\subset} & \mathbb P_{k+2}(D; \mathbb R^3)\ar@<0.4ex>[r]^-{\dev\grad}\ar@<0.4ex>[l]^-{\boldsymbol \pi_{RT}} & \mathbb P_{k+1}(D; \mathbb T)\ar@<0.4ex>[r]^-{\sym\curl}\ar@<0.4ex>[l]^-{\cdot\boldsymbol x}  & \mathbb P_k(D; \mathbb S) \ar@<0.4ex>[r]^-{\div\div}\ar@<0.4ex>[l]^-{\times\boldsymbol x} & \mathbb P_{k-2}(D)  \ar@<0.4ex>[r]^-{} \ar@<0.4ex>[l]^-{\boldsymbol x\boldsymbol x^{\intercal}}
& 0 \ar@<0.4ex>[l]^-{\supset} }.
$}
\end{equation*}

It follows from~\eqref{eq:piRTprop} and the complex~\eqref{eq:divdivKoszulcomplex3dPoly} that
\[
\mathbb P_{k}(D; \mathbb R^3)=(\mathbb P_{k-1}(D; \mathbb T)\cdot\boldsymbol x)\oplus\boldsymbol{RT} \quad k\geq 1.
\]
We then move to the space $\mathbb P_{k+1}(D; \mathbb T)$ and $\mathbb P_{k}(D; \mathbb S)$.

\begin{lemma}\label{lem:symmpolyspacedirectsum}
We have the decompositions
\begin{equation*}
\mathbb P_{k}(D; \mathbb T)=(\mathbb P_{k-1}(D; \mathbb S)\times\boldsymbol x)\oplus\dev\grad\mathbb P_{k+1}(D; \mathbb R^3) \quad k\geq 1,
\end{equation*}
and 
\begin{equation*}
\mathbb P_{k}(D; \mathbb S)=\sym\curl \,\mathbb P_{k+1}(D; \mathbb T) \oplus \boldsymbol x\boldsymbol x^{\intercal} \mathbb P_{k-2}(D)\quad k\geq 2.
\end{equation*}
\end{lemma}

When $D\subset\mathbb R^2$,
the divdiv polynomial complex in two dimensions
\begin{equation}\label{eq:divdivcomplexPoly}
\boldsymbol{RT}\xrightarrow{\subset} \mathbb P_{k+1}(D;\mathbb R^2)\xrightarrow{\sym\curl} \mathbb P_k(D;\mathbb S) \xrightarrow{\div\div} \mathbb P_{k-2}(D)\xrightarrow{}0
\end{equation}
has been proved in~\cite{ChenHuang2020} and used to construct a finite element divdiv complex in two dimensions.

\section{A Conforming Virtual Element Hessian Complex}\label{sec:vemhesscomplex}
In this section we shall construct virtual element and finite element spaces and obtain a discrete Hessian complex ($k\geq3$):
\begin{equation}\label{eq:sechesscomplexvemfem}
\mathbb P_1(\Omega)\xrightarrow{\subset} W_h\xrightarrow{\nabla^2}\boldsymbol \Sigma_h\xrightarrow{\curl} \boldsymbol V_h \xrightarrow{\div} \mathcal Q_h\xrightarrow{}0,
\end{equation}
where
\begin{itemize}
 \item $W_h$ is an $H^2(\Omega)$-conforming virtual element space containing piecewise $\mathbb P_{k+2}$ polynomials;
 \item $\boldsymbol \Sigma_h$ is an $\boldsymbol H(\curl,\Omega; \mathbb S)$-conforming virtual element space containing piecewise $\mathbb P_k$ polynomials; 
 \item $\boldsymbol V_h$ is an $\boldsymbol H(\div, \Omega; \mathbb T)$-conforming finite element space  containing piecewise $\mathbb P_{k-1}$ polynomials;
 \item $\mathcal Q_h$ is piecewise $\mathbb P_{k-2}(\mathbb R^3)$ polynomial which is obviously conforming to $L^2(\Omega)$.
\end{itemize}
The domain $\Omega$ is decomposed into a triangulation $\mathcal T_h$ consisting of tetrahedrons. That is each element $K\in \mathcal T_h$ is a tetrahedron. Extension to general polyhedral meshes will be explored in a future work. 

In~\cite{HuLiang2020}, a finite element Hessian complex has been constructed and the lowest polynomial degree for $(W_h, \boldsymbol \Sigma_h, \boldsymbol V_h, \mathcal Q_h)$ is $(9,7,6,5)$ and ours is $(5,3,2,1)$ but with a few additional virtual shape functions in $W_h$ and $\boldsymbol \Sigma_h$. 

For each element $K\in\mathcal{T}_h$, denote by $\boldsymbol{n}_K $ the
unit outward normal vector to $\partial K$,  which will be abbreviated as $\boldsymbol{n}$.
Let $\mathcal{F}_h$, $\mathcal{E}_h$ and $\mathcal{V}_h$ be the union of all faces, edges and vertices
of the partition $\mathcal {T}_h$, respectively.
For any $F\in\mathcal{F}_h$,
fix a unit normal vector $\boldsymbol{n}_F$.
For any $e\in\mathcal{E}_h$,
fix a unit tangent vector $\boldsymbol{t}_e$ and two unit normal vectors $\boldsymbol{n}_{e,1}$ and $\boldsymbol{n}_{e,2}$, which will be abbreviated as $\boldsymbol{n}_{1}$ and $\boldsymbol{n}_{2}$ without causing any confusions.
For $K$ being a polyhedron, denote by $\mathcal{F}(K)$, $\mathcal{E}(K)$ and $\mathcal{V}(K)$ the set of all faces, edges and vertices of $K$, respectively. 
For any $F\in\mathcal{F}_h$, let $\mathcal{E}(F)$ and $\mathcal{V}(F)$ be the set of all edges and vertices of $F$, respectively. For each $e\in\mathcal{E}(F)$, denote by $\boldsymbol n_{F,e}$ the unit vector
being parallel to $F$ and outward normal to $\partial F$.

\subsection{$H(\div)$-conforming element for trace-free tensors}

For an integer $k\geq 3$, we choose $\mathbb P_{k-1}(K;\mathbb T)$ as the shape function space. Its trace $\boldsymbol v \boldsymbol n$ on each face $F$ is in $\mathbb P_{k-1}(F; \mathbb R^3)$. In the classic $H(\div)$ element for vector functions, such trace can be determined by the face moments $\int_{F}(\boldsymbol v\boldsymbol n) \cdot\boldsymbol q$ for $\boldsymbol q\in \mathbb P_{k-1}(F; \mathbb R^3)$. For the tensor polynomial with additional structure, e.g., here is the trace-free, face moments cannot reflect to this property. One fix is to introduce the nodal continuity of each component of the tensor so that the structure of the tensor is utilized.

For any $F\in\mathcal F(K)$, 
let $\mathbb P_{k-1,2}^{\perp}(F)\subseteq \mathbb P_{k-1}(F)$ be the $L^2$-orthogonal complement space of $\mathbb P_{2}(F)$ in $\mathbb P_{k-1}(F)$ with respect to the $L^2$-inner product $(\cdot, \cdot)_F$ on face $F$.
Denote by $\mathbb P_{k-1,2}^{\perp}(F; \mathbb{R}^d)$ the vector version of $\mathbb P_{k-1,2}^{\perp}(F)$ with $d=2,3$.
Let $\mathbb P_{k-2,{\rm RT}}^{\perp}(K;\mathbb R^3)\subseteq \mathbb P_{k-2}(K;\mathbb R^3)$ be the $L^2$-orthogonal complement space of $\boldsymbol{RT}$ in $\mathbb P_{k-2}(K;\mathbb R^3)$ with respect to the inner product $(\cdot, \cdot)_K$.

\begin{lemma}\label{lem:HdivTunisolvencepre}
Let $F\in\mathcal F(K)$ be a triangular face and $v\in \mathbb P_{k-1}(F)$. If
\[
v(a_1)=v(a_2)=v(a_3)=0,\quad (v, q)_F=0\quad\forall~q\in \mathbb P_1(F)\oplus \mathbb P_{k-1,2}^{\perp}(F)
\]
with $a_1, a_2$ and $a_3$ being the vertices of triangle $F$, then $v=0$.
\end{lemma}
\begin{proof}
Since $v\in \mathbb P_{k-1}(F)$ and $(v, q)_F=0$ for all $q\in \mathbb P_{k-1,2}^{\perp}(F)$,
we get $v\in \mathbb P_{2}(F)$.
Let $(\lambda_1, \lambda_2, \lambda_3)$ be the barycentric coordinate of point $\boldsymbol x$ with respect to $F$. 
Since $v(a_1)=v(a_2)=v(a_3)=0$, we have
$v=c_1\lambda_2\lambda_3+c_2\lambda_3\lambda_1+c_3\lambda_1\lambda_2$, where $c_1, c_2$ and $c_3$ are constants. Now taking $q=\lambda_i$ with $i=1,2,3$, we obtain
\[
\frac{1}{60}|F|
\begin{pmatrix}
1 & 2 & 2\\
2 & 1 & 2\\
2 & 2 & 1
\end{pmatrix}\begin{pmatrix}
c_1\\
c_2\\
c_3
\end{pmatrix}=\begin{pmatrix}
0\\
0\\
0
\end{pmatrix}.
\]
Noting that the coefficient matrix is invertible,  it follows $c_1=c_2=c_3=0$.
\end{proof}

Next we use the $H(\div; \mathbb T)$ polynomial bubble function space introduced in~\cite{HuLiang2020} to characterize the interior part. Denote by 
$$
\mathbb B_{k-1}(K; \mathbb T) := \mathbb P_{k-1}(K; \mathbb T)\cap \boldsymbol H_0(\div, K; \mathbb T),
$$
where $\boldsymbol H_0(\div, K; \mathbb T):=\{\boldsymbol\tau\in\boldsymbol H(\div, K; \mathbb T): \boldsymbol\tau\boldsymbol n|_{\partial K}=\boldsymbol0\}$.
In~\cite{HuLiang2020}, a constructive characterization of $\mathbb B_{k-1}(K; \mathbb T) $ is given by 
\begin{equation}\label{eq:divbubble}
\mathbb B_{k-1}(K; \mathbb T) =\sum_{i=1}^{4} \sum_{1 \leq j<l \leq 4 \atop j, l \neq i} \lambda_{j} \lambda_{l}\mathbb P_{k-3}(K) \boldsymbol{n}_{i} \boldsymbol{t}_{j, l}^{\intercal},
\end{equation}
where $(\lambda_1, \lambda_2, \lambda_3, \lambda_4)$ is the barycentric coordinate of point $\boldsymbol x$ with respect to $K$, and $\boldsymbol{t}_{j, l}:=\boldsymbol{x}_{l}-\boldsymbol{x}_{j}$ with the set of vertices $\mathcal V(K):=\{\boldsymbol{x}_{1}, \boldsymbol{x}_{2}, \boldsymbol{x}_{3}, \boldsymbol{x}_{4}\}$.
That is on each face use the normal vector and an edge vector to form a traceless matrix and extend to the whole element by the scalar edge bubble function. 
It was proved in \cite{HuLiang2020} that
\begin{equation}\label{eq:divbubbledivonto}
\div\mathbb B_{k-1}(K; \mathbb T)=\mathbb P_{k-2,{\rm RT}}^{\perp}(K;\mathbb R^3).
\end{equation}

The sum in \eqref{eq:divbubble}, however, is not a direct sum. We present a refined characterization of the bubble function below.
\begin{lemma}
We have 
\begin{equation}\label{eq:bubbleTcharac}
\mathbb B_{k-1}(K; \mathbb T) =\sum_{i=1}^{4} \bigoplus_{1 \leq j<l \leq 4 \atop j, l \neq i} \lambda_{j} \lambda_{l}\mathbb P_{k-3}^{F_{ijl}}(K) \boldsymbol{n}_{i} \boldsymbol{t}_{j, l}^{\intercal} \oplus \sum_{i=1}^{4}  \sum_{1 \leq j<l \leq 4 \atop j, l \neq i}b_{F_i}\mathbb P_{k-4}(K) \boldsymbol{n}_{i} \boldsymbol{t}_{j, l}^{\intercal},
\end{equation}
where $b_{F_i}$ is the cubic face bubble function corresponding to face $F_i$ and
$$
\mathbb P_{k-3}^{F_{ijl}}(K):=
\textrm{span}\big\{\lambda_i^{\alpha_1}\lambda_j^{\alpha_2}\lambda_l^{\alpha_3}: \alpha_1,\alpha_2,\alpha_3\in\mathbb N, \alpha_1+\alpha_2+\alpha_3=k-3\big\}.
$$
\end{lemma}
\begin{proof}
By $\lambda_{j} \lambda_{l}\mathbb P_{k-3}(K)=\lambda_{j} \lambda_{l}\mathbb P_{k-3}^{F_{ijl}}(K)+b_{F_i}\mathbb P_{k-4}(K)$, it follows from \eqref{eq:divbubble} that
$$
\mathbb B_{k-1}(K; \mathbb T) =\sum_{i=1}^{4} \sum_{1 \leq j<l \leq 4 \atop j, l \neq i} \lambda_{j} \lambda_{l}\mathbb P_{k-3}^{F_{ijl}}(K) \boldsymbol{n}_{i} \boldsymbol{t}_{j, l}^{\intercal}  +  \sum_{i=1}^{4}  \sum_{1 \leq j<l \leq 4 \atop j, l \neq i}b_{F_i}\mathbb P_{k-4}(K) \boldsymbol{n}_{i} \boldsymbol{t}_{j, l}^{\intercal}.
$$


Next we prove 
\begin{align*}
&\quad \sum_{1 \leq j<l \leq 4 \atop j, l \neq i} \lambda_{j} \lambda_{l}\mathbb P_{k-3}^{F_{ijl}}(K)\boldsymbol{t}_{j, l} +  \sum_{1 \leq j<l \leq 4 \atop j, l \neq i}b_{F_i}\mathbb P_{k-4}(K) \boldsymbol{t}_{j, l} \\
&=\bigoplus_{1 \leq j<l \leq 4 \atop j, l \neq i} \lambda_{j} \lambda_{l}\mathbb P_{k-3}^{F_{ijl}}(K)\boldsymbol{t}_{j, l} \oplus \sum_{1 \leq j<l \leq 4 \atop j, l \neq i}b_{F_i}\mathbb P_{k-4}(K) \boldsymbol{t}_{j, l}.
\end{align*}
Consider $i=4$.
Assume there exist $q_{jl}\in\mathbb P_{k-3}^{F_{4jl}}(K)$ and $p_{jl}\in\mathbb P_{k-4}(K)$ for $1\leq j<l\leq 3$ such that
$$
\lambda_{1} \lambda_{2}q_{12}\boldsymbol{t}_{1, 2}+\lambda_{1} \lambda_{3}q_{13}\boldsymbol{t}_{1, 3}+\lambda_{2}\lambda_{3}q_{23}\boldsymbol{t}_{2, 3} + b_{F_4}p_{12}\boldsymbol{t}_{1, 2}+b_{F_4}p_{13}\boldsymbol{t}_{1, 3}+b_{F_4}p_{23}\boldsymbol{t}_{2, 3} =\boldsymbol0.
$$
Hence 
$$
\resizebox{\textwidth}{!}{$
(\lambda_{1}\lambda_{2}q_{12}+\lambda_{1} \lambda_{3}q_{13} + b_{F_4}(p_{12}+p_{13}))\boldsymbol{t}_{1, 2}+(\lambda_{2}\lambda_{3}q_{23}+\lambda_{1} \lambda_{3}q_{13} + b_{F_4}(p_{23}+p_{13}))\boldsymbol{t}_{2, 3} =\boldsymbol0,
$}
$$ 
which implies
$$
\lambda_{2}q_{12}+\lambda_{3}q_{13}+\lambda_{2}\lambda_{3}(p_{12}+p_{13})=0, \quad \lambda_{2}q_{23}+\lambda_{1}q_{13}+\lambda_{1}\lambda_{2}(p_{23}+p_{13})=0.
$$
Therefore $q_{12}=q_{13}=q_{23}=0$, as required.
\end{proof}

By \eqref{eq:bubbleTcharac}, we have
\begin{align*}
\dim \mathbb B_{k-1}(K; \mathbb T)&=12{k-1\choose2}+8{k-1\choose3} =\frac{2}{3}(k-1)(k-2)(2k+3)\\
&=\frac{2}{3}(2k^3-3k^2-5k+6),
\end{align*}
$$
\dim(\mathbb B_{k-1}(K; \mathbb T)\cap\ker(\div))=\frac{1}{6}k(k+1)(5k-17)+8=\frac{1}{6}(5k^3-12k^2-17k)+8.
$$

Now we define an $H(\div)$-conforming finite element for trace-free tensors with $k\geq3$.
Take $\mathbb P_{k-1}(K;\mathbb T)$ as the space of shape functions.
The degrees of freedom are given by
\begin{align}
\boldsymbol v (\delta) & \quad\forall~\delta\in \mathcal V(K), \label{HdivTfem3ddof1}\\
(\boldsymbol v\boldsymbol  n, \boldsymbol q)_F & \quad\forall~\boldsymbol q\in\mathbb P_{1}(F;\mathbb R^3)\oplus \mathbb P_{k-1,2}^{\perp}(F;\mathbb R^3),  F\in\mathcal F(K),\label{HdivTfem3ddof2}\\
(\boldsymbol v, \boldsymbol q)_K & \quad\forall~\boldsymbol q\in \dev\grad\mathbb P_{k-2}(K;\mathbb R^3)\oplus(\mathbb B_{k-1}(K; \mathbb T)\cap\ker(\div)). \label{HdivTfem3ddof3}
\end{align}
We can also replace the degrees of freedom~\eqref{HdivTfem3ddof3} by 
\begin{equation}\label{eq:divbubbledof}
(\boldsymbol v, \boldsymbol q)_K  \quad\forall~\boldsymbol q\in \mathbb B_{k-1}(K; \mathbb T). 
\end{equation}
Thanks to the explicit formulation of bubble functions~\eqref{eq:divbubble}, the implementation using~\eqref{eq:divbubbledof} will be easier. On the other hand,~\eqref{HdivTfem3ddof3} will be helpful when defining discrete spaces for $\boldsymbol H(\curl,K; \mathbb S)$.

\begin{lemma}\label{lem:unisovlenHdivTfem}
The degrees of freedom~\eqref{HdivTfem3ddof1}-\eqref{HdivTfem3ddof3} are unisolvent for $\mathbb P_{k-1}(K;\mathbb T)$.
\end{lemma}
\begin{proof}
First of all the number of the degrees of freedom~\eqref{HdivTfem3ddof1}-\eqref{HdivTfem3ddof3} is
$$
32 + 36 +[6k(k+1)-72] + \big[\frac{1}{2}(k^3-k)-4\big] + \frac{1}{6}(5k^3-12k^2-17k)+8=\frac{4}{3}k(k+1)(k+2),
$$
which equals to $\dim\mathbb P_{k-1}(K;\mathbb T)$.

Take any $\boldsymbol v\in\mathbb P_{k-1}(K;\mathbb T)$ and
suppose all the degrees of freedom~\eqref{HdivTfem3ddof1}-\eqref{HdivTfem3ddof3} vanish. 
Applying Lemma~\ref{lem:HdivTunisolvencepre} to each component of $\boldsymbol v\boldsymbol n$,  we get $\boldsymbol v\in \mathbb B_{k-1}(K; \mathbb T)$. It follows from the integration by parts and the first part of the degrees of freedom~\eqref{HdivTfem3ddof3} that $\div\boldsymbol v=\boldsymbol0$, i.e., $\boldsymbol v\in \mathbb B_{k-1}(K; \mathbb T)\cap\ker(\div)$. Finally we arrive at $\boldsymbol v=\boldsymbol0$ by using the second part of the degrees of freedom~\eqref{HdivTfem3ddof3}.
\end{proof}

The global finite element space is
\begin{align*}
\boldsymbol V_h:=\{\boldsymbol v_h\in \boldsymbol H(\div, \Omega;\mathbb T) :&\, \boldsymbol v_h|_K\in \mathbb P_{k-1}(K;\mathbb T) \; \forall~K\in\mathcal T_h,  \textrm{ all degrees of } \\
&\qquad\qquad\qquad \quad\,\textrm{ freedom are single-valued}\},
\end{align*}
For $\boldsymbol v\in \boldsymbol V_h$, by Lemma~\ref{lem:HdivTunisolvencepre}, the trace $\boldsymbol v \boldsymbol n|_F \in \mathbb P_{k-1}(F;\mathbb R^3)$ is determined uniquely by the degree of freedom~\eqref{HdivTfem3ddof1}-\eqref{HdivTfem3ddof2}. Therefore $\boldsymbol V_h\subset \boldsymbol H(\div, \Omega;\mathbb T)$ is a conforming finite element space.

\subsection{$H^2$-conforming virtual element}\label{sec:H2P5VEM}
To define an $H^2$-conforming virtual element in three dimensions, we shall adapt two dimensional $H^2$-conforming virtual elements constructed in~\cite{BrezziMarini2013,Antonietti;Da-Veiga;Scacchi;Verani:2016Virtual} and three dimensional  $C^1$ virtual element in~\cite{BeiraodaVeigaDassiRusso2020}. 

Define an $H^2$-conforming virtual element space on tetrahedron $K$
\begin{align*}
\widetilde W(K):= \{v\in H^2(K):&  \Delta^2v\in \mathbb P_{k-2}(K), 
\,\textrm{both } v|_{\partial K} \textrm{ and } \nabla v|_{\partial K} \textrm{ are continuous}, \\
& \quad\; v|_F\in \mathbb P_{k+2}(F),  \partial_{n}v|_F\in\mathbb P_{k+1}(F) \textrm{ for each } F\in\mathcal F(K)\}.
\end{align*}
The space of degrees of freedom $\mathcal N(K)$ consists of
\begin{align}
v (\delta), \nabla v (\delta), \nabla^2v (\delta) & \quad\forall~\delta\in \mathcal V(K), \label{H2vemk513ddof1}\\
(v, q)_e & \quad \forall~q\in\mathbb P_{k-4}(e), e\in \mathcal E(K), \label{H2vemk513ddof2}\\
(\partial_{n_i}v, q)_e & \quad \forall~q\in\mathbb P_{k-3}(e), e\in \mathcal E(K),  i=1, 2, \label{H2vemk513ddof3}\\
(v, q)_F & \quad \forall~q\in\mathbb P_{k-4}(F), F\in \mathcal F(K),\label{H2vemk513ddof4}\\
(\partial_{n}v, q)_F & \quad \forall~q\in\mathbb P_{k-2}(F), F\in \mathcal F(K),\label{H2vemk513ddof5}\\
(v, q)_K  &\quad \forall~q\in\mathbb P_{k-2}(K).
\label{H2vemk513ddof6}
\end{align}

The space $\widetilde W(K)$ is not empty as $\mathbb P_{k+2}(K)\subset \widetilde W(K)$. Its dimension is, however, not so clear from the definition. There is a compatible condition given implicitly in the definition of the local space $\widetilde W(K)$. As the trace of a function in $H^2(K)$, the boundary value $v|_{\partial K}$ and $\partial_n v|_{\partial K}$ are compatible in the sense that $\nabla v|_{F} = \nabla_F v + (\partial_n v)|_F \boldsymbol n_F$ should be continuous on edges~\cite[Theorem 5]{BuffaGeymonat2001}. The degree of freedom $\nabla^2v(\delta)$ is also questionable for a function $v\in H^2(K)$ only. In the classic finite element space, this is not an issue as shape functions are polynomials. 

For a more rigorous verification of unisolvence, we introduce data space 
\begin{align*}
\mathcal D(K) = \{ & (f, v_0, \boldsymbol v_1, \boldsymbol v_2, u_0^e, \boldsymbol u_1^e, u_0^F, u_1^F):  \, f\in \mathbb P_{k-2}(K), v_0 \in \mathbb P_0(\mathcal V(K)),  \\ 
& \qquad \boldsymbol v_1 \in \mathbb P_0(\mathcal V(K),\mathbb R^3), \boldsymbol v_2\in \mathbb P_0(\mathcal V(K),\mathbb S), u_0^e \in \mathbb P_{k-4}(\mathcal E(K)), \\
& \qquad \boldsymbol u_1^e \in \mathbb P_{k-3}(\mathcal E(K),\mathbb R^2), 
u_0^F\in \mathbb P_{k-4}(\mathcal F(K)), u_1^F\in \mathbb P_{k-2}(\mathcal F(K))\}.
\end{align*}
Obviously $\dim \mathcal D(K) = \dim \mathcal N(K)$.
For function $v\in \widetilde W(K)\cap C^2(K)$, the mapping
$$(\Delta^2 v, v (\delta), \nabla v (\delta), \nabla^2v (\delta), Q_{k-4}^ev, Q_{k-3}^e(\partial_{n_i}v), Q_{k-4}^Fv, Q_{k-2}^F(\partial_{n}v)),$$
for all $\delta \in \mathcal V(K)$, $e\in \mathcal E(K)$ and $F\in \mathcal F(K)$, is from $\widetilde W(K)\cap C^2(K)\to \mathcal D(K)$. 

Let $\mathbb P_k(\partial K)$ be the function space which is continuous on the boundary $\partial K$ and its restriction to each face is a polynomial of degree at most $k$. Given a data $(f, v_0, \boldsymbol v_1, \boldsymbol v_2, u_0^e, \boldsymbol u_1^e, u_0^F, u_1^F) \in\mathcal D(K)$, using $(v_0, \boldsymbol v_1, \boldsymbol v_2, u_0^e, \boldsymbol u_1^e, u_0^F)$, we can determine a $\mathbb P_{k+2}(F)$ Argyris element~\cite{ArgyrisFriedScharpf1968,BrennerSung2005} and consequently  define a function $g_1\in \mathbb P_{k+2}(\partial K)$.  Similarly using $(\boldsymbol v_1, \boldsymbol v_2, \boldsymbol u_1^e, u_1^F)$, we can determine a $\mathbb P_{k+1}(F)$ Hermite element~\cite{Ciarlet1978} and consequently a function $g_2\in \mathbb P_{k+1}(\partial K)$. By the unisolvence of the Argyris element and Hermite element in two dimensions, we know $(g_1,g_2)$ is uniquely determined by $(v_0, \boldsymbol v_1, \boldsymbol v_2, u_0^e, \boldsymbol u_1^e, u_0^F, u_1^F)$ and $\left(g_2|_F\boldsymbol n_F+\nabla_F(g_1|_F)\right)|_e$ is single-valued across each edge $e\in\mathcal E(K)$.

Given data $(f,g_1,g_2)$, we consider the biharmonic equation with Dirichlet boundary condition
\begin{equation}\label{eq:biharmonic1}
\Delta^2 v = f \text{ in } K, \quad v = g_1, \partial_n v = g_2 \text{ on } \partial K.
\end{equation}
As $g_1, g_2$ are compatible in the sense $g_2 \boldsymbol n + \nabla_{\partial K}(g_1)\in \mathbb P_{k+1}(\partial K;\mathbb R^3)$ with $\mathbb P_{k+1}(\partial K;\mathbb R^3)$ being the vector version of $\mathbb P_{k+1}(\partial K)$, by the trace theorem of $H^2(K)$ on polyhedral domains~\cite[Theorem 5]{BuffaGeymonat2001}, there exists $v^b\in H^2(K)$ such that
\[
v^b|_{\partial K}=g_1,\quad \partial_{n}v^b|_{\partial K}=g_2.
\]
Indeed $v^b$ can be chosen as a polynomial in $\mathbb P_{\max\{k+1,9\}}(K)$ using the $\mathcal C^1$ finite element constructed in~\cite{Zhang:2009family}. 
Then consider the biharmonic equation with the homogenous boundary condition
\begin{equation*}
\Delta^2 v^0 = f -\Delta^2v^b \text{ in } K, \quad v^0 = 0, \partial v^0 = 0 \text{ on } \partial K.
\end{equation*}
The existence and uniqueness of $v^0$ is guaranteed by the Lax-Milligram lemma. 
Setting $v=v^b+v^0$ gives a solution to~\eqref{eq:biharmonic1}. The uniqueness of the solution to~\eqref{eq:biharmonic1} is trivial. 

Therefore we have constructed an embedding operator $\mathcal L: \mathcal D(K) \to \widetilde W(K)$ and $\mathcal L$ is injective. We shall choose $$W(K) = \mathcal L(\mathcal D(K))$$ and by construction $\mathcal L: \mathcal D(K) \to W(K)$ is a bijection. 
Functions in $W(K)$ are defined as solutions to \eqref{eq:biharmonic1} which may still not be smooth enough to take nodal values of the Hessian. 

To be consistent with finite element notation, we still use the form $\nabla^2v(\delta)$ but understand it with the help of $\mathcal L$. For $v\in W(K)$, $\mathcal L^{-1}v = (f, v_0, \boldsymbol v_1, \boldsymbol v_2, u_0^e, \boldsymbol u_1^e, u_0^F, u_1^F) \in \mathcal D(K)$. We define $\nabla^2v(\delta)\in W'(K)$ by
\begin{equation}\label{eq:hessianvertex}
\nabla^2v(\delta) := \boldsymbol v_2.
\end{equation}
That is we understand $\nabla^2v$ as a functional defined on $W(K)$ which will match the vertex value of the hessian if $v$ is smooth enough. Other degrees of freedom~\eqref{H2vemk513ddof1}-\eqref{H2vemk513ddof5} can be understood in a similar fashion. The interior moment~\eqref{H2vemk513ddof6} keeps unchanged and the relation of~\eqref{H2vemk513ddof6} and $f \in \mathcal L^{-1}v$ is discussed below. 

\begin{lemma}
The degrees of freedom~\eqref{H2vemk513ddof1}-\eqref{H2vemk513ddof6} are unisolvent for $W(K)$.
\end{lemma}
\begin{proof}
First of all $\dim W(K) = \dim \mathcal N(K)=\frac{1}{6}(k^3+24k^2+35k+60)$. Take any $v\in W(K)$ and
suppose all the degrees of freedom~\eqref{H2vemk513ddof1}-\eqref{H2vemk513ddof6} vanish. 
By the unisolvence of the Argyris element and Hermite element in two dimensions, we have $v\in H_0^2(K)$.
It follows from the integration by parts that
\[
\|\nabla^2v\|_{0,K}^2=(\Delta^2v,v)_{0,K}=0,
\]
as $\Delta^2 v\in \mathbb P_{k-2}(K)$ and the vanishing degree of freedom~\eqref{H2vemk513ddof6}. Thus $v=0$. 
\end{proof}

As $\dim \mathbb P_{k+2}(K) = \frac{1}{6}(k^3+12k^2+47k+60)$, there are $2k(k-1)$ shape functions in $W(K)$ are non-polynomials and thus are treated as virtual. The $L^2$-projection of $\nabla^2 v$ to $\mathbb P_k(K,\mathbb S)$ can be computed by degrees of freedom using the following Green's identity~\cite{Chen;Huang:2020Finite}: for $\boldsymbol \tau \in \mathbb P_k(K,\mathbb S)$ and $v\in W(K)$,
\begin{align}
(\nabla^2v, \boldsymbol \tau)_K&= (\div\div\, \boldsymbol \tau, v)_K +\sum_{F\in\mathcal F(K)}\sum_{e\in\mathcal E(F)}(\boldsymbol n_{F,e}^{\intercal}\boldsymbol \tau \boldsymbol n, v)_e \notag\\
&\quad + \sum_{F\in\mathcal F(K)}\left[(\boldsymbol  n^{\intercal}\boldsymbol \tau\boldsymbol  n, \partial_n v)_{F} -  ( 2\div_F(\boldsymbol \tau\boldsymbol n)+\partial_n (\boldsymbol  n^{\intercal}\boldsymbol \tau\boldsymbol  n), v)_F\right]. \notag
\end{align}
As $\div\div\, \boldsymbol \tau\in \mathbb P_{k-2}(K)$, the first term can be computed by~\eqref{H2vemk513ddof6}. On the boundary, $v|_F$ is a $\mathbb P_{k+2}(F)$ Argyris element, and $\partial_n v|_{F}$ is a $\mathbb P_{k+1}(F)$ Hermite element and thus all boundary terms are computable.  In particular by choosing $\boldsymbol \tau \in \nabla^2\mathbb P_{k+2}(K)$, we can compute an $H^2$-projection of $v$ to $\mathbb P_{k+2}(K)$, that is $\Pi^{K}v\in\mathbb P_{k+2}(K)$ is determined by
\begin{align}\label{eq:H2projection}
(\nabla^2\Pi^{K}v, \nabla^2q)_K&=(\nabla^2v, \nabla^2q)_K\quad\forall~q\in\mathbb P_{k+2}(K),
\\
\label{eq:H2projectionP1} (\Pi^{K}v, q)_K&=(v, q)_K\qquad\quad\,\forall~q\in\mathbb P_{1}(K).
\end{align}
We have the following properties of $\Pi^{K}$. Obviously $\Pi^{K}$ is a projector, i.e.,
\[
\Pi^{K}q=q\quad\forall~q\in\mathbb P_{k+2}(K).
\]
By the standard Bramlbe-Hilbert lemma, we have 
\begin{equation}\label{eq:Pikerrestimate}
h_K^i|v-\Pi^{K}v|_{i,K}\lesssim h_K^2\inf_{q\in\mathbb P_{k+2}(K)}|v-q|_{2,K}\quad\forall~v\in H^2(K),  i=0,1,2.
\end{equation}

\begin{remark}\label{rm:macro}\rm
The $C^1$ macro-element on the Alfeld split
in~\cite{FuGuzmanNeilan2018, Alfeld1984,LaiSchumaker2007} has the same degrees of freedom on boundary as~\eqref{H2vemk513ddof1}-\eqref{H2vemk513ddof5}.
We can construct a conforming macro-element Hessian complex on the Alfeld split following the approach in this paper. Here we present the lowest order $C^1$ macro-element, i.e. $k=3$.
For any tetrahedron $K$, let Alfeld split $\mathcal T_A(K)$ be the set of the four subtetrahedra obtained by connecting $\boldsymbol x_K$
to each of the vertices of $K$, where $\boldsymbol x_K$ is the barycenter of $K$.
The shape function space of the lowest order $C^1$ macro-element on the Alfeld split
in~\cite{FuGuzmanNeilan2018, Alfeld1984,LaiSchumaker2007} is given by
$$
W_A(K):=\{v\in H^2(K): v|_{K'}\in \mathbb P_5(K') \textrm{ for each } K'\in \mathcal T_A(K) \}.
$$
And the degrees of freedom are
\begin{align}
v (\delta), \nabla v (\delta), \nabla^2v (\delta) & \quad\forall~\delta\in \mathcal V(K), \label{H2macfemk513ddof1}\\
\int_e\partial_{n_i}v\dd s & \quad \forall~e\in \mathcal E(K),  i=1, 2, \label{H2macfemk513ddof2}\\
(\partial_{n}v, q)_F & \quad \forall~q\in\mathbb P_1(F), F\in \mathcal F(K),\label{H2macfemk513ddof3}\\
(\nabla v, \nabla q)_K  &\quad \forall~q\in \mathring{W}_A(K),
\label{H2macfemk513ddof4}
\end{align}
where $\mathring{W}_A(K):=\{v\in W_A(K): \textrm{all the degrees of freedom~\eqref{H2macfemk513ddof1}-\eqref{H2macfemk513ddof3} vanish}\}$.
$\hfill\Box$
\end{remark}

For any $F\in\mathcal F(K)$, both $v|_F$ and $\partial_{n_F}v|_F$ are determined by the degrees of freedom~\eqref{H2vemk513ddof1}-\eqref{H2vemk513ddof5} on the face $F$. Thus 
we can define the $H^2$-conforming virtual element space
\begin{align*}
W_h:=\{v_h\in H^2(\Omega):&\, v_h|_K\in W(K) \textrm{ for each }K\in\mathcal T_h,  \textrm{ all degrees of } \\
&\qquad\quad\;\;\;\;\textrm{ freedom~\eqref{H2vemk513ddof1}-\eqref{H2vemk513ddof6} are single-valued}\}.
\end{align*}

Let $I_h^{\Delta}: H^4(\Omega)\to W_h$ be the nodal interpolation operator with respect to the degrees of freedom~\eqref{H2vemk513ddof1}-\eqref{H2vemk513ddof6}. For each tetrahedron $K$, by the scaling argument and the norm equivalence on the finite dimensional spaces (cf. \cite[Section~3.1]{Ciarlet1978}), it holds
\begin{equation}\label{eq:IhDeltaerrestimate}
h_K^i|v-I_h^{\Delta}v|_{i,K}\lesssim h_K^{k+2}|v|_{k+2,K}\quad\forall~v\in H^{k+2}(\Omega),  i=0,1,2.
\end{equation}
Here we take the advantage that the element is a tetrahedron and by transferring back to the reference element, one can show the constant in \eqref{eq:IhDeltaerrestimate} depends only on the shape regularity of the element.

\subsection{Trace complexes}\label{sec:tracecomplex}
We have the following trace complexes
\begin{equation}\label{eq:tracecomplex1}
\begin{array}{c}
\xymatrix{
\boldsymbol{a \cdot x} +  b \ar[d] \ar[r]^-{\subset} & v \ar[d] \ar[r]^-{\hess}
                & \boldsymbol  \tau \ar[d]   \ar[r]^-{\curl} & \ar[d]\boldsymbol  \sigma \ar[r]^{\div} & \boldsymbol p \\
\boldsymbol a_F \cdot \boldsymbol  x_F +  b_F \ar[r]^{\subset} &   v|_F \ar[r]^{\nabla^2_F}
                & \Pi_F \boldsymbol  \tau \Pi_F \ar[r]^{\mathrm{rot}_F} &  \boldsymbol  n^{\intercal} \boldsymbol  \sigma \Pi_F \ar[r]^{}& \boldsymbol0    }
\end{array},
\end{equation}
where $b_F:=\boldsymbol a\cdot\boldsymbol n(\boldsymbol x\cdot\boldsymbol n)|_F+b$, and
\begin{equation}\label{eq:tracecomplex2}
\begin{array}{c}
\xymatrix{
\boldsymbol  a\cdot \boldsymbol  x + \boldsymbol  b \ar[d] \ar[r]^-{\subset} & v \ar[d] \ar[r]^-{\hess}
                & \boldsymbol  \tau \ar[d]   \ar[r]^-{\curl} & \ar[d]\boldsymbol  \sigma \ar[r]^{\div} & \boldsymbol p \\
\boldsymbol a\cdot\boldsymbol n \ar[r]^{\subset} & \partial_n v|_F \ar[r]^{{\rm grad}_F}
                & \boldsymbol n^{\intercal}\boldsymbol  \tau \Pi_F   \ar[r]^{{\rm rot}_F} &  \boldsymbol  n^{\intercal} \boldsymbol  \sigma  \boldsymbol  n \ar[r]^{}& 0    }
\end{array}.
\end{equation}
In \eqref{eq:tracecomplex1} and \eqref{eq:tracecomplex2}, on the bottom of the diagram, all functions are evaluated on one face $F$. We present the concrete form instead of trace operators of Sobolev spaces as we will work mostly on polynomial functions when restricting to faces.

The trace complexes will motivate the correct continuity and degree of freedom on edges and faces. For example, the $2\times 2$ symmetric matrix $\Pi_F \boldsymbol  \tau \Pi_F\in H({\rm rot}_F, F,\mathbb S)$ and the vector $\boldsymbol n^{\intercal}\boldsymbol  \tau \Pi_F\in H({\rm rot}_F, F,\mathbb R^2)$ imply the tangential continuity of $\boldsymbol \tau\boldsymbol t$ on edges. The face moments for $\boldsymbol n^{\intercal}\boldsymbol  \tau \Pi_F$ will come from that of the N\'ed\'elec element. The face moments for $\Pi_F \boldsymbol  \tau \Pi_F$ will be based on the decomposition build-in the polynomial complex~\eqref{eq:divdivcomplexPoly}.

One important relation is the commutative diagram build-in the trace complex. For example, the third block of \eqref{eq:tracecomplex1} and \eqref{eq:tracecomplex2} implies $\rot _F(\boldsymbol \tau \Pi_F) = (\curl \, \boldsymbol\tau)\boldsymbol  n|_F$ which can be verified easily by definition.

As $\div_F(\boldsymbol\tau\times\boldsymbol n) = \rot _F(\boldsymbol \tau \Pi_F)$, i.e., $\div_F$ is a rotation of $\rot_F$, the trace $\boldsymbol \tau \times \boldsymbol n\in H(\div_F,F)$ and conclusion for $\boldsymbol \tau \times \boldsymbol n$ can be transfer to $\boldsymbol \tau \Pi_F$ and vice verse. 

\subsection{$H(\curl)$-conforming element for symmetric tensors}
Motivated by the decomposition~\eqref{eq:hesspolyspacedecomp2}, we take the space of shape functions
\[
\boldsymbol\Sigma(K):=\nabla^2W(K)\oplus \sym(\mathbb P_{k-1}(K;\mathbb T)\times\boldsymbol x).
\]
The degrees of freedom are given by
\begin{align}
\curl \,  \boldsymbol\tau (\delta) & \quad\forall~\delta\in \mathcal V(K), \label{HcurlSvemk313ddof0}\\
\boldsymbol\tau (\delta) & \quad\forall~\delta\in \mathcal V(K), \label{HcurlSvemk313ddof1}\\
(\boldsymbol \tau\boldsymbol  t, \boldsymbol q)_e & \quad\forall~\boldsymbol q\in\mathbb P_{k-2}(e;\mathbb R^3),  e\in\mathcal E(K),\label{HcurlSvemk313ddof2}\\
(\Pi_F\boldsymbol \tau\Pi_F, \boldsymbol q)_F & \quad\forall~\boldsymbol q\in\mathbb P_0(F,\mathbb S)\oplus\sym\nabla_F^{\perp}\mathbb P_{k-1,2}^{\perp}(F;\mathbb R^2)\oplus \boldsymbol x\boldsymbol x^{\intercal}\mathbb P_{k-4}(F), \notag\\
& \quad\;\;\, F\in\mathcal F(K),\label{HcurlSvemk313ddof3}\\
(\boldsymbol n^{\intercal}\boldsymbol \tau\Pi_F, \boldsymbol q)_F & \quad\forall~\boldsymbol q\in\mathbb P_0(F,\mathbb R^2)\oplus\nabla_F^{\perp}\mathbb P_{k-1,2}^{\perp}(F)\oplus\mathbb P_{k-2}(F)\boldsymbol x, F\in\mathcal F(K),\label{HcurlSvemk313ddof4}\\
(\curl \,  \boldsymbol\tau, \boldsymbol q)_K & \quad\forall~\boldsymbol q\in \mathbb B_{k-1}(K; \mathbb T)\cap\ker(\div), \label{HcurlSvemk313ddof5}\\
(\boldsymbol \tau, \boldsymbol x\boldsymbol x^{\intercal}q)_K & \quad\forall~ q\in \mathbb P_{k-2}(K).
\label{HcurlSvemk313ddof6}
\end{align}

From the decomposition~\eqref{eq:hesspolyspacedecomp2}, we know that $\mathbb P_k(K;\mathbb S)\subset \boldsymbol\Sigma(K)$. The dimension of the space is
$$
\dim\boldsymbol\Sigma(K)=\dim W(K)-4+\dim\sym(\mathbb P_{k-1}(K;\mathbb T)\times\boldsymbol x)=k^3+8k^2+9k+6.
$$
The number of the degrees of freedom~\eqref{HcurlSvemk313ddof0}-\eqref{HcurlSvemk313ddof6} is
\begin{align*}
32 + 24 + 18(k-1) + (6k^2-6k-24) + (4k^2-16)& \\
+\frac{1}{6}(5k^3-12k^2-17k+48)+\frac{1}{6}(k^3-k)&=k^3+8k^2+9k+6,
\end{align*}
which agrees with $\dim\boldsymbol\Sigma(K)$. In~\eqref{HcurlSvemk313ddof3}-\eqref{HcurlSvemk313ddof4} we separate the trace $\boldsymbol \tau\Pi_F$ into the tangential-tangential part $\Pi_F\boldsymbol \tau\Pi_F$ and the tangential-normal part $\boldsymbol n^{\intercal}\boldsymbol \tau\Pi_F$. 
Most of the shape functions in $\boldsymbol\Sigma(K)$ are polynomials except $2k(k-1)$ non-polynomial ones in the form $\nabla^2 v$ for some $v\in W(K)$ and $\nabla^2v(\delta)$ should be understood in the sense of~\eqref{eq:hessianvertex}.

Although there are non-polynomial shape functions, the trace $\boldsymbol \tau\times\boldsymbol n$ on each face is always polynomial and determined by~\eqref{HcurlSvemk313ddof0}-\eqref{HcurlSvemk313ddof4}. 
\begin{lemma}\label{lem:HcurlSvemk31conforming}
For each $F\in\mathcal F(K)$ and any $\boldsymbol\tau\in\boldsymbol\Sigma(K)$, $\boldsymbol\tau\times\boldsymbol n|_F\in \mathbb P_k(F;\mathbb M)$ is determined by the degrees of freedom~\eqref{HcurlSvemk313ddof0}-\eqref{HcurlSvemk313ddof4} on face $F$.
\end{lemma}
\begin{proof}
First of all, we show although $\boldsymbol \tau \in \boldsymbol\Sigma(K)$ may be from a virtual element space, its trace $\boldsymbol\tau\times\boldsymbol n|_F\in \mathbb P_k(F;\mathbb M)$.  To see this, it suffices to  check $(\nabla^2 v)\Pi_F$ for $v\in W(K)$. Using notation in Section~\ref{sec:matvec}, it is straightforward to verify that
\begin{equation*}
\Pi_F \nabla^2v \Pi_F = \nabla_F^2 (v|_F), \quad \boldsymbol n \cdot \nabla^2v\Pi_F= \nabla_F(\partial_n v|_{F}).
\end{equation*}
As $v|_F\in \mathbb P_{k+2}(F)$ and $\partial_n v|_{F}\in \mathbb P_{k+1}(F)$ are polynomials, $\boldsymbol\tau\times\boldsymbol n|_F$ is a polynomial of degree $k$. 

Assume all the degrees of freedom~\eqref{HcurlSvemk313ddof0}-\eqref{HcurlSvemk313ddof4} on face $F$ are zeros. We are going to prove this polynomial is vanished. 
The vanishing degrees of freedom~\eqref{HcurlSvemk313ddof1}-\eqref{HcurlSvemk313ddof2} imply $\boldsymbol \tau\boldsymbol t|_e = \boldsymbol0$ for every $e\in \partial F$ as $\boldsymbol \tau\boldsymbol t|_e\in \mathbb P_k(e;\mathbb R^3)$. Then $\boldsymbol\tau\times\boldsymbol n|_F\in \boldsymbol H_0(\div_F, F)$. Using the integration by parts and the vanishing degrees of freedom~\eqref{HcurlSvemk313ddof3}-\eqref{HcurlSvemk313ddof4}, we obtain
\[
(\div_F(\boldsymbol\tau\times\boldsymbol n), \boldsymbol q)_F= (\boldsymbol\tau\times\boldsymbol n, \grad_F\boldsymbol q)_F = 0\quad\forall~\boldsymbol q\in\mathbb P_{1}(F;\mathbb R^3)\oplus \mathbb P_{k-1,2}^{\perp}(F;\mathbb R^3).
\]
Using the relation $-\div_F(\boldsymbol\tau\times\boldsymbol n) = (\curl \, \boldsymbol\tau)\boldsymbol  n|_F\in \mathbb P_{k-1}(F; \mathbb R^3)$ and the vanishing degree of freedom~\eqref{HcurlSvemk313ddof0}, we know $\div_F(\boldsymbol\tau\times\boldsymbol n)(\delta) = 0$ for all $\delta \in \mathcal V(F)$. 
Applying Lemma~\ref{lem:HdivTunisolvencepre},  we acquire $\div_F(\boldsymbol\tau\times\boldsymbol n )=\boldsymbol0$ which is equivalent to $\rot_F (\boldsymbol \tau\Pi_F) = 0$.

The tangential component $\boldsymbol \tau\Pi_F$ can be further decomposed into two components: the tangential-tangential part $\Pi_F \boldsymbol \tau\Pi_F$ and the tangential-normal part $\boldsymbol n^{\intercal}\boldsymbol \tau\Pi_F$. Noting that $\boldsymbol n^{\intercal}\boldsymbol \tau\Pi_F\in H_0(\rot_F,F)\cap\mathbb P_k(F;\mathbb R^2)$ and $\rot_F (\boldsymbol n^{\intercal}\boldsymbol \tau\Pi_F) = 0$, which implies $\boldsymbol n^{\intercal}\boldsymbol \tau\Pi_F\bot \nabla_F^{\bot} H^1(F)$. 
We get from the vanishing degrees of freedom~\eqref{HcurlSvemk313ddof4} that 
\[
(\boldsymbol n^{\intercal}\boldsymbol \tau\Pi_F, \boldsymbol q)_F=0\quad\forall~\boldsymbol q\in\mathbb P_{k-1}(F;\mathbb R^2),
\]
where we use the decomposition $\mathbb P_{k-1}(F; \mathbb R^2) = \nabla_F^{\bot}\mathbb P_k(F) \oplus \boldsymbol x\mathbb P_{k-2}(F)$ which is a two dimensional version of \eqref{eq:RTdecomposition}. Due to the unisolvence of the second-type N\'ed\'elec element~\cite{Nedelec:1986family}, we get $\boldsymbol n^{\intercal}\boldsymbol \tau\Pi_F=\boldsymbol0$. 

For the tangential-tangential part, as $\Pi_F\boldsymbol \tau\Pi_F\in \mathbb P_k(F; \mathbb S)$, by the Hessian complex~\eqref{eq:hesscomplex2dPoly} in two dimensions, there exists $w_F\in\mathbb P_{k+2}(F)$ such that $\Pi_F\boldsymbol \tau \Pi_F =\nabla_F^2w_F$ and $w_F(\delta)=0$ for each $\delta\in\mathcal V(F)$.
Then we get from the vanishing degrees of freedom~\eqref{HcurlSvemk313ddof1}-\eqref{HcurlSvemk313ddof2}
that
\[
\nabla_F^2w_F(\delta)=\boldsymbol0 \quad\forall~\delta\in \mathcal V(F),
\]
\[
(\partial_t(\nabla_Fw_F), \boldsymbol q)_e=0 \quad\forall\boldsymbol q\in\mathbb P_{k-2}(e;\mathbb R^3),  e\in \mathcal E(F),
\]
which indicate $\partial_t(\nabla_Fw_F)|_e=\boldsymbol0$ for each $e\in \mathcal E(F)$. As a result $w_F\in  H_0^2(F)$. 
Due to the vanishing degrees of freedom~\eqref{HcurlSvemk313ddof3},
$$
(w_F, \div_F\div_F(\boldsymbol x\boldsymbol x^{\intercal}q))_F=(\nabla_F^2w_F, \boldsymbol x\boldsymbol x^{\intercal}q)_F=0\quad\forall~q\in\mathbb P_{k-4}(F).
$$
Therefore by $\div_F\div_F(\boldsymbol x\boldsymbol x^{\intercal}\mathbb P_{k-4}(F))=\mathbb P_{k-4}(F)$, cf. \eqref{eq:divdivcomplexPoly},
and the unisolvence of the Argyris element, it follows that $w_F=0$.
\end{proof}

To show the unisolvence, we adapt the unisolvence proof of three dimensional $H(\curl)$-conforming virtual element in~\cite{Da-Veiga;Brezzi;Marini;Russo:2016curl-conforming}. We take the advantage of the fact that $K$ is a tetrahedron and $\curl \boldsymbol\Sigma(K)$ is polynomial. The  approach of using local problems is troublesome as for symmetric matrices, the well-posedness of $\curl-\div$ system with non-homogenous Dirichlet boundary condition is unclear. A crucial and missing part is the characterization of the trace space of $\boldsymbol H(\curl, \Omega; \mathbb S)$. 

\begin{lemma}\label{lem:unisovlenHcurlSvemk31}
The degrees of freedom~\eqref{HcurlSvemk313ddof0}-\eqref{HcurlSvemk313ddof6} are unisolvent for $\boldsymbol\Sigma(K)$.
\end{lemma}
\begin{proof}
Take any $\boldsymbol\tau\in\boldsymbol\Sigma(K)$ and
suppose all the degrees of freedom~\eqref{HcurlSvemk313ddof0}-\eqref{HcurlSvemk313ddof6} vanish. We are going to prove $\boldsymbol \tau = \boldsymbol0$.

With vanishing degrees of freedom~\eqref{HcurlSvemk313ddof0}-\eqref{HcurlSvemk313ddof4}, we have proved that $\boldsymbol \tau \in \boldsymbol H_0(\curl, K;\mathbb S)$. Then $\curl \, \boldsymbol \tau \in \mathbb B_{k-1}(K,\mathbb T)\cap \ker(\div)$, together with the vanishing degree of freedom~\eqref{HcurlSvemk313ddof5} implies $\curl \, \boldsymbol \tau = \boldsymbol0$. 

Using integration by parts, with $\boldsymbol \tau \times \boldsymbol n |_{\partial K}= \boldsymbol0$ and $\curl \boldsymbol \tau = \boldsymbol0$, 
\begin{equation}\label{eq:symcurl}
(\boldsymbol \tau, \sym \curl \, \boldsymbol \sigma)_K = (\curl \, \boldsymbol \tau, \boldsymbol \sigma)_K + (\boldsymbol \tau \times \boldsymbol n, \boldsymbol \sigma)_{\partial K},
\end{equation}
we conclude that $\boldsymbol \tau \bot  \sym \curl \, \boldsymbol \sigma$ for any $\boldsymbol \sigma \in \boldsymbol H( \sym \curl; \mathbb M)$. 

Use the fact $\div\div \boldsymbol \tau \in \mathbb P_{k-2}(K)$ and $\div\div: \boldsymbol x\boldsymbol x^{\intercal }\mathbb P_{k-2}(K) \to \mathbb P_{k-2}(K)$ is a bijection, cf. Lemma~\ref{lem:symmpolyspacedirectsum}, we can find a polynomial $\boldsymbol x\boldsymbol x^{\intercal }q$ with $q\in\mathbb P_{k-2}(K)$ such that $\div\div (\boldsymbol\tau - \boldsymbol x\boldsymbol x^{\intercal }q) = 0$ and thus $\boldsymbol \tau  = \boldsymbol x\boldsymbol x^{\intercal }q + \sym \curl \, \boldsymbol \sigma$ for some $\boldsymbol \sigma \in \boldsymbol H( \sym \curl; \mathbb M)$. 

Then by the vanishing degree of freedom~\eqref{HcurlSvemk313ddof6},
$$
(\boldsymbol \tau, \boldsymbol \tau)_K = (\boldsymbol \tau,  \boldsymbol x\boldsymbol x^{\intercal }q +  \sym \curl \, \boldsymbol \sigma)_K = 0,
$$
which implies $\boldsymbol \tau = \boldsymbol0$. 
\end{proof}

We now discuss how to compute the $L^2$-projection of an element $\boldsymbol \tau\in \boldsymbol\Sigma(K)$ to $\mathbb P_k(K;\mathbb S)$. By Lemma~\ref{lem:HcurlSvemk31conforming}, we can determine the piecewise polynomial $\boldsymbol\tau\times \boldsymbol n$ on the boundary and $(\curl \, \boldsymbol\tau)\boldsymbol n|_F$. Together with~\eqref{HcurlSvemk313ddof5}, $\curl \, \boldsymbol \tau\in \mathbb P_{k-1}(K;\mathbb T)$ is determined. Then, using~\eqref{eq:symcurl}, we can compute the $L^2$-projection to the subspace $\sym \curl \, \mathbb P_{k+1}(K; \mathbb T)$. Use the degree of freedom~\eqref{HcurlSvemk313ddof6}, we can compute the $L^2$-projection to the subspace $\boldsymbol x\boldsymbol x^{\intercal} \mathbb P_{k-2}(K)$. Finally, recalling that $\mathbb P_k(K; \mathbb S) = \boldsymbol x\boldsymbol x^{\intercal} \mathbb P_{k-2}(K) \oplus \sym \curl \mathbb P_{k+1}(K; \mathbb S)$, the $L^2$-projection to $\mathbb P_k(K; \mathbb S)$ will be obtained by combining the projection to each subspace and an orthogonalization step.

Define the global finite element space
\begin{align*}
\boldsymbol\Sigma_h:=\{\boldsymbol\tau_h\in \boldsymbol L^2(\Omega;\mathbb S):&\, \boldsymbol\tau_h|_K\in \boldsymbol\Sigma(K) \quad\forall~K\in\mathcal T_h,  \textrm{ all degrees of } \\
&\quad\quad\quad\qquad\;\;\;\;\textrm{ freedom are single-valued}\}.
\end{align*}
It follows from Lemma~\ref{lem:HcurlSvemk31conforming} that $\boldsymbol\Sigma_h\subset \boldsymbol H(\curl, \Omega; \mathbb S)$.

For any sufficiently smooth and symmetric tensor $\boldsymbol\tau$ defined on tetrahedron $K$, let $\boldsymbol I_K^c\boldsymbol\tau\in\boldsymbol \Sigma(K)$ be the nodal interpolation of $\boldsymbol\tau$ based on the degrees of freedom~\eqref{HcurlSvemk313ddof0}-\eqref{HcurlSvemk313ddof6}. We have
\[
\boldsymbol I_K^c\boldsymbol\tau=\boldsymbol\tau\quad\forall~\boldsymbol\tau\in\boldsymbol \Sigma(K),
\]
and by the scaling argument and the norm equivalence on the finite dimensional spaces (cf. \cite[Section~3.1]{Ciarlet1978})
\begin{equation}\label{eq:Ikcprop1}
\|\boldsymbol\tau-\boldsymbol I_K^c\boldsymbol\tau\|_{0,K}+h_K\|\curl(\boldsymbol\tau-\boldsymbol I_K^c\boldsymbol\tau)\|_{0,K}\lesssim h_K^{k+1}|\boldsymbol\tau|_{k+1,K}\quad\forall~\boldsymbol\tau\in\boldsymbol H^{k+1}(K; \mathbb S).
\end{equation}
Again by transferring back to the reference tetrahedron, one can show the constant in \eqref{eq:Ikcprop1} depends only on the shape regularity of the tetrahedron.
For any sufficiently smooth and symmetric tensor $\boldsymbol\tau$ defined on $\Omega$, let $\boldsymbol I_h^c\boldsymbol\tau\in\boldsymbol \Sigma_h$ be defined by $(\boldsymbol I_h^c\boldsymbol\tau)|_K:=\boldsymbol I_K^c(\boldsymbol\tau|_K)$ for each $K\in\mathcal T_h$.

If $\boldsymbol\tau\in\boldsymbol H^1(K; \mathbb S)$ satisfying $\curl \,\boldsymbol\tau \in \mathbb P_{k-1}(K;\mathbb T)$, 
due to Lemma~5.38 in~\cite{Monk2003} and Lemma~4.7 in~\cite{AmroucheBernardiDaugeGirault1998}, the interpolation $\boldsymbol I_K^c\boldsymbol\tau$ is well-defined, and it follows from the integration by parts and Lemma~\ref{lem:unisovlenHdivTfem} that 
\begin{equation}\label{eq:curlIhc}
\curl(\boldsymbol I_K^c\boldsymbol\tau)=\curl \,  \boldsymbol\tau.
\end{equation}
Moreover, by the scaling argument we have
\begin{equation}\label{eq:Ihcerror}
\|\boldsymbol\tau-\boldsymbol I_K^c\boldsymbol\tau\|_{0,K}\lesssim h_K|\boldsymbol\tau|_{1,K}.
\end{equation}

\begin{remark}\label{rm:SigmaA}\rm
We can define an $H(\curl)$-conforming macro-element for symmetric tensors. Let $W_A(K)$ be the $H^2$-conforming macro-element defined in Remark \ref{rm:macro}. Take the space of shape functions
\[
\boldsymbol\Sigma_A(K):=\nabla^2W_A(K)\oplus \sym(\mathbb P_2(K;\mathbb T)\times\boldsymbol x).
\]
And the degrees of freedom are given by
\begin{align}
\curl \,  \boldsymbol\tau (\delta) & \quad\forall~\delta\in \mathcal V(K), \label{HcurlSmacfemk313ddof0}\\
\boldsymbol\tau (\delta) & \quad\forall~\delta\in \mathcal V(K), \label{HcurlSmacfemk313ddof1}\\
(\boldsymbol \tau\boldsymbol  t, \boldsymbol q)_e & \quad\forall~\boldsymbol q\in\mathbb P_{1}(e;\mathbb R^3),  e\in\mathcal E(K),\label{HcurlSmacfemk313ddof2}\\
(\boldsymbol n\times\boldsymbol \tau\times\boldsymbol n, \boldsymbol q)_F & \quad\forall~\boldsymbol q\in\mathbb P_0(F,\mathbb S), F\in\mathcal F(K),\label{HcurlSmacfemk313ddof3}\\
(\boldsymbol n^{\intercal}\boldsymbol \tau\Pi_F, \boldsymbol q)_F & \quad\forall~\boldsymbol q\in\mathbb P_0(F,\mathbb R^2)\oplus\mathbb P_1(F)\boldsymbol x, F\in\mathcal F(K),\label{HcurlSmacfemk313ddof4}\\
(\curl \,  \boldsymbol\tau, \boldsymbol q)_K & \quad\forall~\boldsymbol q\in \mathbb B_2(K; \mathbb T)\cap\ker(\div), \label{HcurlSmacfemk313ddof5}\\
(\boldsymbol \tau, q\boldsymbol I)_K & \quad\forall~ q\in \mathring{W}_A(K).
\label{HcurlSmacfemk313ddof6}
\end{align}
The degrees of freedom \eqref{HcurlSmacfemk313ddof0}-\eqref{HcurlSmacfemk313ddof6} are the same as \eqref{HcurlSvemk313ddof0}-\eqref{HcurlSvemk313ddof6} except \eqref{HcurlSmacfemk313ddof6}, which is inspired by \eqref{H2macfemk513ddof4} when defining $W_A$. One advantage of using the macro-element is that the shape functions are piecewise polynomial and thus no need to compute the $L^2$-projection. $\Box$
\end{remark}

\subsection{Discrete conforming Hessian complex}
In this subsection we will prove the sequence~\eqref{eq:sechesscomplexvemfem} forms a discrete Hessian complex in three dimensions.

The polynomial space for $L^2(\Omega)$ is simply discontinuous $\mathbb P_{k-2}$ space
\[
\mathcal Q_h:=\{\boldsymbol q_h\in \boldsymbol L^2(\Omega;\mathbb R^3) : \boldsymbol q_h|_K\in \mathbb P_{k-2}(K;\mathbb R^3) \quad\forall~K\in\mathcal T_h\}.
\]

\begin{lemma}
It holds
\begin{equation}\label{eq:divontodiscrete}
\div\boldsymbol V_h=\mathcal Q_h.
\end{equation}
\end{lemma}
\begin{proof}
It is apparent that $\div\boldsymbol V_h\subseteq\mathcal Q_h$. Conversely taking any $\boldsymbol p_h\in\mathcal Q_h$, by~\eqref{eq:divH1TontoL2} there exists $\boldsymbol v\in \boldsymbol H^1(\Omega; \mathbb T)$ such that $\div\boldsymbol v=\boldsymbol p_h$. Choose $\boldsymbol v_1\in \boldsymbol V_h$ determined by
\begin{align*}
\boldsymbol v_1 (\delta)&= \boldsymbol0, \\
(\boldsymbol v_1\boldsymbol  n, \boldsymbol q)_F&=(\boldsymbol v\boldsymbol  n, \boldsymbol q)_F  \quad\;\,\forall~\boldsymbol q\in\mathbb P_{1}(F;\mathbb R^3)\oplus \mathbb P_{k-1,2}^{\perp}(F;\mathbb R^3), 
\\
(\boldsymbol v_1, \boldsymbol q)_K&=(\boldsymbol v, \boldsymbol q)_K  \quad\;\;\;\;\forall~\boldsymbol q\in \dev\grad\mathbb P_{k-2}(K;\mathbb R^3)\oplus(\mathbb B_{k-1}(K; \mathbb T)\cap\ker(\div)) 
\end{align*}
for each $\delta\in \mathcal V_h$, $F\in\mathcal F_h$ and $K\in\mathcal T_h$.
It follows from the integration by parts that
\[
(\div(\boldsymbol v-\boldsymbol v_1), \boldsymbol q)_K=0\quad\forall~\boldsymbol q\in \mathbb P_{1}(K;\mathbb R^3), K\in\mathcal T_h,
\]
which means $\div(\boldsymbol v-\boldsymbol v_1)|_K\in\mathbb P_{k-2,{\rm RT}}^{\perp}(K;\mathbb R^3)$.
Employing \eqref{eq:divbubbledivonto}, there exists $\boldsymbol v_2\in\boldsymbol V_h$ such that $\div(\boldsymbol v-\boldsymbol v_1)=\div\boldsymbol v_2$.
Therefore $\div\boldsymbol v_h=\div\boldsymbol v=\boldsymbol q_h$ by setting $\boldsymbol v_h=\boldsymbol v_1+\boldsymbol v_2$.
\end{proof}

\begin{lemma}\label{lem:hesscomplexvemk51}
Assume $\Omega$ is a topologically trivial domain.
Then we have the discrete Hessian complex
\begin{equation}\label{eq:hesscomplexncfemk51}
\mathbb P_1(\Omega)\xrightarrow{\subset} W_h\xrightarrow{\nabla^2}\boldsymbol \Sigma_h\xrightarrow{\curl} \boldsymbol V_h \xrightarrow{\div} \mathcal Q_h\xrightarrow{}\boldsymbol0.
\end{equation}
\end{lemma}
\begin{proof}
It is easy to see that~\eqref{eq:hesscomplexncfemk51} is a complex as all discrete spaces are conforming. We check the exactness of this complex. First of all, $W_h\cap\ker(\nabla^2)=\mathbb P_1(\Omega)$. Then
$$
\dim\nabla^2W_h=\dim W_h-4=10\#\mathcal V_h+(3k-7)\#\mathcal E_h+(k^2-3k+3)\#\mathcal F_h+\frac{1}{6}(k^3-k)\#\mathcal T_h-4.
$$

For any $\boldsymbol\tau_h\in\boldsymbol \Sigma_h\cap \ker(\curl)$, there exists $w\in H^2(\Omega)$ satisfying $\boldsymbol\tau_h=\nabla^2w$. On each element $K$, we have $\nabla^2(w|_K)\in \nabla^2W(K)$,  which means $w|_K\in W(K)$. Noting that $\nabla^2w$ is single-valued at each vertex in $\mathcal V_h$. Then $w\in W_h$. This indicates $\boldsymbol \Sigma_h\cap \ker(\curl)=\nabla^2W_h$, and
\begin{align*}
&\quad\dim\curl\boldsymbol \Sigma_h=\dim\boldsymbol \Sigma_h-\dim\nabla^2W_h \\
&= 14\#\mathcal V_h+(3k-3)\#\mathcal E_h+\frac{1}{2}(5k^2-3k-20)\#\mathcal F_h \\
&\quad+(k^3-2k^2-3k+8)\#\mathcal T_h-\dim\nabla^2W_h \\
&= 4\#\mathcal V_h+4\#\mathcal E_h+\frac{1}{2}(3k^2+3k-26)\#\mathcal F_h+\frac{1}{6}(5k^3-12k^2-17k+48)\#\mathcal T_h+4.
\end{align*}
On the other side, it holds from~\eqref{eq:divontodiscrete} that
\begin{align*}
&\dim\boldsymbol V_h\cap\ker(\div)=\dim\boldsymbol V_h-\dim\mathcal Q_h \\
=&8\#\mathcal V_h + \frac{1}{2}(3k^2+3k-18)\#\mathcal F_h + \frac{2}{3}(2k^3-3k^2-5k+6)\#\mathcal T_h - \frac{1}{2}(k^3-k)\#\mathcal T_h \\
=&8\#\mathcal V_h + \frac{1}{2}(3k^2+3k-18)\#\mathcal F_h + \frac{1}{6}(5k^3-12k^2-17k+24)\#\mathcal T_h.    
\end{align*}
Hence we acquire from the Euler's formula that
\[
\dim\boldsymbol V_h\cap\ker(\div)-\dim\curl\boldsymbol \Sigma_h=4(-\#\mathcal T_h+\#\mathcal F_h-\#\mathcal E_h+\#\mathcal V_h-1)=0,
\]
which yields $\boldsymbol V_h\cap\ker(\div)=\curl\boldsymbol \Sigma_h$.
\end{proof}

\begin{remark}\rm
When the topology of $\Omega$ is non-trivial, it is assumed to be captured by the triangulation $\mathcal T_h$. As all discrete spaces are conforming, the co-homology groups defined by the Hessian complex is preserved in the discrete Hessian complex.
\end{remark}

\begin{remark}\rm
When $\Omega$ is a topologically trivial domain,
the following macro-element Hessian complex based on the Alfeld split
\begin{equation*}
\mathbb P_1(\Omega)\xrightarrow{\subset} W_h^A\xrightarrow{\nabla^2}\boldsymbol \Sigma_h^A\xrightarrow{\curl} \boldsymbol V_h \xrightarrow{\div} \mathcal Q_h\xrightarrow{}\boldsymbol0
\end{equation*}
is also exact,
where
\begin{align*}
W_h^A:=\{v_h\in H^2(\Omega):&\, v_h|_K\in W_A(K) \textrm{ for each }K\in\mathcal T_h,  \textrm{ all degrees of } \\
&\qquad\quad\;\;\;\;\;\;\textrm{ freedom~\eqref{H2macfemk513ddof1}-\eqref{H2macfemk513ddof4} are single-valued}\},\\
\boldsymbol\Sigma_h^A:=\{\boldsymbol\tau_h\in \boldsymbol L^2(\Omega;\mathbb S):&\, \boldsymbol\tau_h|_K\in \boldsymbol\Sigma_A(K) \quad\forall~K\in\mathcal T_h,  \textrm{ all degrees of } \\
&\quad\;\;\;\;\;\;\textrm{ freedom~\eqref{HcurlSmacfemk313ddof0}-\eqref{HcurlSmacfemk313ddof6} are single-valued}\}. \qquad\Box
\end{align*}
\end{remark}

\subsection{Discrete Poincar\'e inequaltiy}
Due to the exactness of the discrete Hessian complex, we have the following discrete Poincar\'e inequality. 

\begin{lemma}\label{lem:discretePoincareieqlity}
Assume $\Omega$ is a topologically trivial domain.
For any $\boldsymbol\tau_h\in\boldsymbol \Sigma_h$ satisfying
\[
(\boldsymbol\tau_h,\nabla^2w_h)=0\quad\forall~w_h\in W_h,
\]
it holds the discrete Poincar\'e inequality
\begin{equation}\label{eq:discretePoincareieqlity}
\|\boldsymbol\tau_h\|_0\lesssim \|\curl \, \boldsymbol\tau_h\|_0.
\end{equation}
In general, 
\begin{equation*}
\|\boldsymbol\tau_h\|_0\leq\|\curl \, \boldsymbol\tau_h\|_0+\sup_{w_h\in W_h}\frac{(\boldsymbol\tau_h,\nabla^2w_h)}{\|w_h \|_2}\qquad\forall~\boldsymbol\tau_h\in\boldsymbol \Sigma_h. 
\end{equation*}
\end{lemma}
\begin{proof}
Since $\curl \, \boldsymbol\tau_h\in\boldsymbol H(\div, \Omega;\mathbb T)$, by~\eqref{eq:curlH1SontoDivfree} there exists $\boldsymbol\tau\in \boldsymbol H^1(\Omega;\mathbb S)$ such that 
\begin{equation}\label{eq:regularpoential}
\curl \, \boldsymbol\tau=\curl \, \boldsymbol\tau_h,\quad \|\boldsymbol\tau\|_1\lesssim \|\curl \, \boldsymbol\tau_h\|_0.
\end{equation}
By~\eqref{eq:curlIhc}, we have
\[
\curl(\boldsymbol I_h^c\boldsymbol\tau)=\curl \,  \boldsymbol\tau=\curl \, \boldsymbol\tau_h.
\]
It follows from the complex~\eqref{eq:hesscomplexncfemk51} that $\boldsymbol\tau_h-\boldsymbol I_h^c\boldsymbol\tau\in \nabla^2W_h$.
Hence we obtain from~\eqref{eq:Ihcerror} and~\eqref{eq:regularpoential} that
\[
\|\boldsymbol\tau_h\|_0^2=(\boldsymbol\tau_h,\boldsymbol\tau_h)=(\boldsymbol\tau_h,\boldsymbol I_h^c\boldsymbol\tau)\leq\|\boldsymbol\tau_h\|_0\|\boldsymbol I_h^c\boldsymbol\tau\|_0\lesssim\|\boldsymbol\tau_h\|_0\|\boldsymbol\tau\|_1,
\]
which means~\eqref{eq:discretePoincareieqlity}.

For a general $\boldsymbol \tau_h\in \boldsymbol \Sigma_h$, by the exact sequence~\eqref{eq:hesscomplexncfemk51}, we have the $L^2$-orthogonal Helmholtz decomposition
$$
\boldsymbol \tau_h  = \nabla^2 v_h + \boldsymbol \tau_h^0,
$$
and $\boldsymbol \tau_h^0 \bot \nabla^2 W_h$ whose $L^2$-norm can be controlled by~\eqref{eq:discretePoincareieqlity} $\| \boldsymbol \tau_h^0 \|_0 \lesssim \| \curl \boldsymbol \tau_h^0 \|_0 = \| \curl \boldsymbol \tau_h \|_0$. The first part $\nabla^2 v_h$ is the $L^2$-projection of $\boldsymbol \tau_h$ to $\nabla^2 W_h$ and thus 
$$
\|\nabla^2 v_h\|_0 = \sup_{w_h\in W_h/\mathbb P_1(\Omega)}\frac{(\nabla^2 v_h,\nabla^2w_h)}{\|\nabla^2w_h\|_0} = \sup_{w_h\in W_h/\mathbb P_1(\Omega)}\frac{(\boldsymbol\tau_h,\nabla^2w_h)}{\|\nabla^2w_h\|_0}. 
$$
Then we use Poincar\'e inequality
$$
\| w_h \|_0\lesssim \|\nabla^2w_h\|_0\quad \forall w_h\in W_h/\mathbb P_1(\Omega)
$$
to finish the proof. 
\end{proof}

The discrete Poincar\'e inequality \eqref{eq:discretePoincareieqlity} is the discrete version of Poincar\'e inequality~\eqref{eq:Poincareieqlity}.
The $L^2$-inner product $(\cdot,\cdot)$ and norm $\|\cdot\|_0$ can be changed to an equivalent one and similar results still hold.

\section{Discretization for the Linearized Einstein-Bianchi System}\label{sec:discreteEB}

In this section we will apply the constructed conforming virtual element Hessian complex to discretize the time-independent linearized Einstein-Bianchi system.
 
\subsection{Linearized Einstein-Bianchi system}
Consider the time-independent linearized Einstein-Bianchi system~\cite{Quenneville-Belair2015}: find $\sigma\in H^2(\Omega)$, $\boldsymbol E\in \boldsymbol H(\curl, \Omega; \mathbb S)$ and $\boldsymbol B\in \boldsymbol L^2(\Omega;\mathbb T)$ such that
\begin{align}
\quad\quad(\sigma,\tau)-(\boldsymbol E,\nabla^2\tau)&=0  \qquad\qquad \forall~\tau\in H^2(\Omega), \label{EBsystem1}\\
(\nabla^2\sigma, \boldsymbol v)+(\boldsymbol B, \curl\boldsymbol v)&=(\boldsymbol f, \boldsymbol v)  \qquad \forall~\boldsymbol v\in \boldsymbol H(\curl, \Omega;\mathbb S),  \label{EBsystem2}\\
\;\;\,(\boldsymbol B, \boldsymbol\psi)-(\curl\boldsymbol E, \boldsymbol\psi)&=0 \qquad\qquad \forall~\boldsymbol\psi\in \boldsymbol L^2(\Omega;\mathbb T), \label{EBsystem3}
\end{align}
where $\boldsymbol f\in\boldsymbol L^2(\Omega;\mathbb S)$. Here following~\cite{HuLiang2020,Quenneville-Belair2015} we switch the notation and use $\sigma,\tau$ for functions in $H^2$ and $\boldsymbol E,\boldsymbol v$ for functions in $\boldsymbol H(\curl, \Omega; \mathbb S)$.

To show the well-posedness of the linearized Einstein-Bianchi system~\eqref{EBsystem1}-\eqref{EBsystem3}, we introduce the product space
$$
\mathcal X = H^2(\Omega)\times\boldsymbol H(\curl, \Omega; \mathbb S)\times\boldsymbol L^2(\Omega;\mathbb T)
$$
and the bilinear form 
$A(\cdot, \cdot): \mathcal X\times \mathcal X \to\mathbb R$ as
\[
A(\sigma, \boldsymbol E, \boldsymbol B;\tau, \boldsymbol v, \boldsymbol\psi):=(\sigma,\tau)-(\boldsymbol E,\nabla^2\tau)-(\nabla^2\sigma, \boldsymbol v)-(\boldsymbol B, \curl\boldsymbol v)+(\boldsymbol B, \boldsymbol\psi)-(\curl\boldsymbol E, \boldsymbol\psi).
\]
It is easy to prove the continunity
\begin{equation}\label{eq:bilinearcontinuity}
A(\sigma, \boldsymbol E, \boldsymbol B;\tau, \boldsymbol v, \boldsymbol\psi)\lesssim (\|\sigma\|_2 + \|\boldsymbol E\|_{H(\curl)} + \|\boldsymbol B\|_{0})(\|\tau\|_2 + \|\boldsymbol v\|_{H(\curl)} + \|\boldsymbol\psi\|_{0})
\end{equation}
for any $\sigma,\tau\in H^2(\Omega)$, $\boldsymbol E, \boldsymbol v\in\boldsymbol H(\curl, \Omega; \mathbb S)$ and $\boldsymbol B, \boldsymbol\psi\in \boldsymbol L^2(\Omega;\mathbb T)$.
The well-posedness of~\eqref{EBsystem1}-\eqref{EBsystem3} is then derived from the following inf-sup condition.

\begin{lemma}
For any $\sigma\in H^2(\Omega), \boldsymbol E\in\boldsymbol H(\curl, \Omega; \mathbb S)$ and $\boldsymbol B\in \boldsymbol L^2(\Omega;\mathbb T)$, it holds
\begin{equation}\label{eq:infsup}
\|\sigma\|_2 + \|\boldsymbol E\|_{H(\curl)} + \|\boldsymbol B\|_{0}\lesssim\sup_{(\tau,\boldsymbol v, \boldsymbol \psi)\in \mathcal X}\frac{A(\sigma, \boldsymbol E, \boldsymbol B;\tau, \boldsymbol v, \boldsymbol\psi)}{\|\tau\|_2 + \|\boldsymbol v\|_{H(\curl)} + \|\boldsymbol\psi\|_{0}}.
\end{equation}
\end{lemma}
\begin{proof}
For ease of presentation, let
\[
\alpha=\sup_{(\tau,\boldsymbol v, \boldsymbol \psi)\in \mathcal X}\frac{A(\sigma, \boldsymbol E, \boldsymbol B;\tau, \boldsymbol v, \boldsymbol\psi)}{\|\tau\|_2 + \|\boldsymbol v\|_{H(\curl)} + \|\boldsymbol\psi\|_{0}}.
\]
Then it follows from the Poincar\'e inequality that
\begin{align}
\|\boldsymbol E\|_0&\lesssim \|\curl\boldsymbol E\|_0 + \sup_{\tau\in H^2(\Omega)}\frac{(\boldsymbol E, \nabla^2\tau)}{\|\tau\|_2} \notag\\
&\leq\|\curl\boldsymbol E\|_0 + \|\sigma\|_0 + \sup_{\tau\in H^2(\Omega)}\frac{(\boldsymbol E, \nabla^2\tau)-(\sigma,\tau)}{\|\tau\|_2} \notag\\
&\leq\|\curl\boldsymbol E\|_0 + \|\sigma\|_0 + \alpha. \label{eq:20201102}
\end{align}
On the other side, we have
\[
A\Big(\sigma, \boldsymbol E, \boldsymbol B;\sigma, -\boldsymbol E-\nabla^2\sigma, \frac{1}{2}(\boldsymbol B-\curl\boldsymbol E)\Big)=\|\sigma\|_0^2+|\sigma|_2^2+\frac{1}{2}\|\boldsymbol B\|_0^2+\frac{1}{2}\|\curl\boldsymbol E\|_0^2.
\]
Hence we get from the definition of $\alpha$ and~\eqref{eq:20201102} that
\begin{align*}
&\|\sigma\|_0^2+|\sigma|_2^2+\frac{1}{2}\|\boldsymbol B\|_0^2+\frac{1}{2}\|\curl\boldsymbol E\|_0^2 \\
\leq \, &\alpha(\|\sigma\|_2 + \|\boldsymbol E+\nabla^2\sigma\|_{H(\curl)} + \frac{1}{2}\|\boldsymbol B-\curl\boldsymbol E\|_{0}) \\
\lesssim \, &\alpha(\|\sigma\|_2 + \|\boldsymbol E\|_{0} + \|\curl\boldsymbol E\|_{0}+ \|\boldsymbol B\|_{0}) \\
\lesssim \,&\alpha(\|\sigma\|_2 + \|\curl\boldsymbol E\|_{0} + \|\boldsymbol B\|_{0}) +\alpha^2,
\end{align*}
which yields
\[
\|\sigma\|_2 + \|\curl\boldsymbol E\|_{0} + \|\boldsymbol B\|_{0}\lesssim \alpha.
\]
Finally the inf-sup condition~\eqref{eq:infsup} follows from the last inequality and~\eqref{eq:20201102}.
\end{proof}

As a result of~\eqref{eq:bilinearcontinuity} and the inf-sup condition~\eqref{eq:infsup}, the variational formulation~\eqref{EBsystem1}-\eqref{EBsystem3} of the linearized Einstein-Bianchi system is well-posed, and
\[
\|\sigma\|_2 + \|\boldsymbol E\|_{H(\curl)} + \|\boldsymbol B\|_{0}\lesssim \|\boldsymbol f\|_{(\boldsymbol H(\curl, \Omega; \mathbb S))'}.
\]

It follows from~\eqref{EBsystem3} that $\boldsymbol B=\curl\boldsymbol E$,  which can be eliminated from the system,
so the linearized Einstein-Bianchi system~\eqref{EBsystem1}-\eqref{EBsystem3} is equivalent to
find $\boldsymbol E\in \boldsymbol H(\curl, \Omega; \mathbb S)$ and $\sigma\in H^2(\Omega)$  such that
\begin{align}
a(\boldsymbol E, \boldsymbol v)+b(\boldsymbol v,\nabla^2\sigma)&=(\boldsymbol f, \boldsymbol v) \quad\;\, \forall~\boldsymbol v\in \boldsymbol H(\curl, \Omega; \mathbb S), \label{EBsystem2term1} \\
b(\boldsymbol E,\nabla^2\tau)-c(\sigma,\tau)&=0 \quad\quad\quad\;\; \forall~\tau\in H^2(\Omega), \label{EBsystem2term2}
\end{align}
where
\[
a(\boldsymbol E, \boldsymbol v)=(\curl\boldsymbol E, \curl\boldsymbol v), \quad b(\boldsymbol E,\nabla^2\tau)=(\boldsymbol E,\nabla^2\tau), \quad c(\sigma,\tau)=(\sigma,\tau).
\]
Then the inf-sup condition~\eqref{eq:infsup} is equivalent to
\begin{equation*}
\|\sigma\|_2 + \|\boldsymbol E\|_{H(\curl)}\lesssim\sup\limits_{\tau\in H^2(\Omega)\atop \boldsymbol v\in\boldsymbol H(\curl, \Omega; \mathbb S)}\frac{a(\boldsymbol E, \boldsymbol v)+b(\boldsymbol v,\nabla^2\sigma)+b(\boldsymbol E,\nabla^2\tau)-c(\sigma,\tau)}{\|\tau\|_2 + \|\boldsymbol v\|_{H(\curl)}}
\end{equation*}
for any $\sigma\in H^2(\Omega)$ and $\boldsymbol E\in\boldsymbol H(\curl, \Omega; \mathbb S)$. 

In summary, the simplified EB system \eqref{EBsystem2term1}-\eqref{EBsystem2term2} can be thought of as a generalization of Maxwell equations for $\boldsymbol E\in H(\curl,\Omega;\mathbb R^3)$ to the tensor version $\boldsymbol E\in H(\curl,\Omega;\mathbb S)$. The scalar potential $\sigma$ is also changed from $H^1(\Omega)$ to $H^2(\Omega)$ as the underline complex is changed from the de Rham complex to the Hessian complex.

\subsection{Conforming Discretization}
With conforming subspaces $W_h$ and $\boldsymbol\Sigma_h$, we could directly consider the Galerkin approximation of~\eqref{EBsystem2term1}-\eqref{EBsystem2term2}. However, as pointwise information of functions in virtual element spaces are not available, the $L^2$-inner product $(\cdot,\cdot)$ involved in $b(\cdot,\cdot)$ and $c(\cdot,\cdot)$ are not computable.

\begin{remark}\rm If we use the macro-elements $W_A$ and $\Sigma_A$ defined on the Alfeld split, cf. Remarks \ref{rm:macro} and \ref{rm:SigmaA}, the shape functions are piecewise polynomials and thus $b(\cdot,\cdot)$ and $c(\cdot,\cdot)$ are computable. 
\end{remark}

We will replace them by equivalent and accurate approximations which can be thought of as numerical quadrature. First introduce two stabilizations
\begin{align*}
S_K^0(\sigma, \tau)&:=h_K(\sigma, \tau)_{\partial K}+h_K^3(\partial_n\sigma, \partial_n\tau)_{\partial K} , \\
S_K^1(\boldsymbol E, \boldsymbol v)&:=h_K^2(\curl\boldsymbol E, \curl\boldsymbol v)_K+h_K(\boldsymbol E\times\boldsymbol n, \boldsymbol v\times\boldsymbol n)_{\partial K}, 
\end{align*}
which are computable as all integrands are polynomials. 

\begin{lemma}
For each tetrahedron $K\in\mathcal T_h$, we have
\begin{align}\label{eq:normequiv1}
S_K^0(\tau, \tau) &\eqsim  \|\tau\|_{0,K}^2\quad\forall~\tau\in W(K)\cap\ker(Q_{k-2}^K),\\
\label{eq:normequiv2}
S_K^1(\boldsymbol v, \boldsymbol v)&\eqsim  \|\boldsymbol v\|_{0,K}^2\quad\forall~\boldsymbol v\in\boldsymbol\Sigma(K)\cap\ker(\boldsymbol Q_k^K).
\end{align}
\end{lemma}
\begin{proof}
By the norm equivalence on the finite dimensional spaces and the scaling argument, it is sufficient to prove $S_K^0(\cdot, \cdot)$ and $S_K^1(\cdot, \cdot)$ are squared norms for the spaces $W(K)\cap\ker(Q_{k-2}^K)$ and $\boldsymbol\Sigma(K)\cap\ker(\boldsymbol Q_k^K)$, respectively. Again as the element is a tetrahedron, by transferring back to the reference element, one can show the constants in \eqref{eq:normequiv1} and \eqref{eq:normequiv2} depends only on the shape regularity of the element.

Assume $\tau\in W(K)\cap\ker(Q_{k-2}^K)$ and $S_K^0(\tau, \tau)=0$. Then $\tau\in H_0^2(K)$. 
By the integration by parts and the definition of $W(K)$, it follows
\[
\|\nabla^2\tau\|_{0,K}^2=(\Delta^2\tau, \tau)_K=(\Delta^2\tau, Q_{k-2}^K\tau)_K=0,
\]
which results in $\tau=0$. Hence $S_K^0(\cdot, \cdot)$ is a squared norm for the space $W(K)\cap\ker(Q_{k-2}^K)$.

Assume $\boldsymbol v\in \boldsymbol\Sigma(K)\cap\ker(\boldsymbol Q_k^K)$ and $S_K^1(\boldsymbol v, \boldsymbol v)=0$. Apparently $\boldsymbol v\in \boldsymbol H_0(\curl, \Omega; \mathbb S)\cap\ker(\curl)$. Then there exists $w\in W(K)\cap H_0^2(K)$ satisfying $\boldsymbol v=\nabla^2w$. Since $\boldsymbol v\in\ker(\boldsymbol Q_k^K)$, we get
\[
(w, \div\div \boldsymbol q)_K=(\nabla^2w, \boldsymbol q)_K=0\quad\forall~\boldsymbol q\in\mathbb P_k(K;\mathbb S),
\] 
which together with complex~\eqref{eq:divdivcomplex3dPoly} implies
\[
(w, q)_K=0\quad\forall~q\in\mathbb P_{k-2}(K).
\] 
Therefore $w=0$ and $\boldsymbol v=0$.
\end{proof}

With these two stabilizations, define local bilinear forms
\begin{align*}
b_K(\boldsymbol E, \boldsymbol v) &:=(\boldsymbol Q_k^K\boldsymbol E, \boldsymbol Q_k^K\boldsymbol v)_K+ S_K^1(\boldsymbol E-\boldsymbol Q_k^K\boldsymbol E, \boldsymbol v-\boldsymbol Q_k^K\boldsymbol v), \\
c_K(\sigma,\tau) &:=(\widetilde{Q}_{k+2}^K\sigma, \widetilde{Q}_{k+2}^K\tau)_K+ S_K^0(\sigma-\widetilde{Q}_{k+2}^K\sigma, \tau-\widetilde{Q}_{k+2}^K\tau),
\end{align*}
where $\widetilde{Q}_{k+2}^K\sigma:=Q_{k-2}^K\sigma + (I- Q_{k-2}^K)\Pi^K\sigma$. Recall that $\Pi^K$ is the $H^2$-projection to $\mathbb P_{k+2}(K)$ defined by \eqref{eq:H2projection}-\eqref{eq:H2projectionP1}. The $L^2$-projection $Q^K_{k+2}\sigma$ is not computable but $Q^K_{k-2}\sigma$ is  using the interior moments. Then $\widetilde{Q}_{k+2}^K\sigma$ is to augment $Q^K_{k-2}\sigma$ by the higher degree part from $\Pi^K\sigma$.
It is obvious that
\begin{align}\label{eq:bKconsistency}
b_K(\boldsymbol E, \boldsymbol q) &= (\boldsymbol E, \boldsymbol q)_K\quad\forall~\boldsymbol E\in \boldsymbol H^1(K;\mathbb S)\cup \boldsymbol\Sigma(K), \boldsymbol q\in\mathbb P_k(K;\mathbb S),\\
\label{eq:cKconsistency}
c_K(\sigma, q) &= (\widetilde{Q}_{k+2}^K\sigma, q)_K\quad\forall~\sigma\in H^2(K)\cap W(K), q\in\mathbb P_{k+2}(K).
\end{align}
And we obtain from~\eqref{eq:normequiv1} and~\eqref{eq:normequiv2} that
\begin{align}\label{eq:normequiv3}
b_K(\boldsymbol v, \boldsymbol v)\eqsim  \|\boldsymbol v\|_{0,K}^2\quad\forall~\boldsymbol v\in\boldsymbol \Sigma(K),\\
\label{eq:normequiv4}
c_K(\tau,\tau)\eqsim \|\tau\|_{0,K}^2\quad\forall~\tau\in W(K).
\end{align}
Then we have from the Cauchy-Schwarz inequality that
\begin{align}
\label{eq:bKcontinuity}
b_K(\boldsymbol E, \boldsymbol v) &\lesssim  \|\boldsymbol E\|_{0,K}\|\boldsymbol v\|_{0,K}\quad\forall~\boldsymbol E, \boldsymbol v\in\boldsymbol \Sigma(K),\\
\label{eq:cKcontinuity}
c_K(\sigma,\tau) &\lesssim \|\sigma\|_{0,K}\|\tau\|_{0,K}\quad\forall~\sigma,\tau\in W(K).
\end{align}

We propose the following conforming mixed virtual element method for the variational formulation~\eqref{EBsystem2term1}-\eqref{EBsystem2term2}: find
$\boldsymbol E_h\in \boldsymbol\Sigma_h$ and $\sigma_h\in W_h$  such that
\begin{align}
a(\boldsymbol E_h, \boldsymbol v_h)+b_h(\boldsymbol v_h,\nabla^2\sigma_h)&=(\boldsymbol f, \boldsymbol Q_h\boldsymbol v_h) \quad\;\, \forall~\boldsymbol v_h\in \boldsymbol\Sigma_h, \label{mvem1} \\
b_h(\boldsymbol E_h,\nabla^2\tau_h)-c_h(\sigma_h,\tau_h)&=0 \quad\quad\quad\quad\;\;\;\; \forall~\tau_h\in W_h, \label{mvem2}
\end{align}
where $\boldsymbol Q_h\boldsymbol v_h\in \boldsymbol L^2(\Omega;\mathbb S)$ is given by $(\boldsymbol Q_h\boldsymbol v_h)|_K:=\boldsymbol Q_k^K(\boldsymbol v_h|_K)$ for each $K\in\mathcal T_h$ and
\[
b_h(\boldsymbol E_h,\nabla^2\tau_h):=\sum_{K\in\mathcal T_h}b_K(\boldsymbol E_h|_K, \nabla^2\tau_h|_K),\quad c_h(\sigma_h,\tau_h):=\sum_{K\in\mathcal T_h}c_K(\sigma_h|_K, \tau_h|_K).
\]

For any $\boldsymbol E_h, \boldsymbol v_h\in \boldsymbol\Sigma_h$ and $\sigma_h, \tau_h\in W_h$, it follows from~\eqref{eq:bKcontinuity}-\eqref{eq:cKcontinuity} that
\begin{equation*}
A_h(\boldsymbol E_h, \sigma_h; \boldsymbol v_h, \tau_h)\leq (\|\boldsymbol E_h\|_{H(\curl)}+\|\sigma_h\|_2)(\|\boldsymbol v_h\|_{H(\curl)}+\|\tau_h\|_2),
\end{equation*}
where
\[
A_h(\boldsymbol E_h, \sigma_h; \boldsymbol v_h, \tau_h):=a(\boldsymbol E_h, \boldsymbol v_h)+b_h(\boldsymbol v_h,\nabla^2\sigma_h)+b_h(\boldsymbol E_h,\nabla^2\tau_h)-c_h(\sigma_h,\tau_h).
\]

Following the proof of Lemma~\ref{lem:discretePoincareieqlity}, we will have 
\begin{align}
\|\boldsymbol v_h\|_0&\lesssim\|\curl\boldsymbol v_h\|_0+\sup_{\tau_h\in W_h}\frac{b_h(\boldsymbol v_h, \nabla^2\tau_h)}{\|\tau_h\|_2}. \label{eq:discretePoincareieqlity3}
\end{align}

We then prove the discrete inf-sup condition.
\begin{lemma}
For any $\boldsymbol E_h\in\boldsymbol\Sigma_h$ and $\sigma_h\in W_h$, it holds
\begin{equation}\label{eq:discreteinfsup}
\|\boldsymbol E_h\|_{H(\curl)}+\|\sigma_h\|_2\lesssim \sup\limits_{\boldsymbol v_h\in\boldsymbol\Sigma_h\atop \tau_h\in W_h}\frac{A_h(\boldsymbol E_h, \sigma_h; \boldsymbol v_h, \tau_h)}{\|\boldsymbol v_h\|_{H(\curl)}+\|\tau_h\|_2}.
\end{equation}
\end{lemma}
\begin{proof}
For ease of presentation, let
\[
\alpha=\sup\limits_{\boldsymbol v_h\in\boldsymbol\Sigma_h\atop \tau_h\in W_h}\frac{A_h(\boldsymbol E_h, \sigma_h; \boldsymbol v_h, \tau_h)}{\|\boldsymbol v_h\|_{H(\curl)}+\|\tau_h\|_2}.
\]
Since
\begin{align*}
\sup_{\tau_h\in W_h}\frac{b_h(\boldsymbol E_h, \nabla^2\tau_h)}{\|\tau_h\|_2}&=\sup_{\tau_h\in W_h}\frac{b_h(\boldsymbol E_h,\nabla^2\tau_h)- c_h(\sigma_h,\tau_h) + c_h(\sigma_h,\tau_h)}{\|\tau_h\|_2}  \\
& \lesssim \|\sigma_h\|_0 + \alpha,
\end{align*}
we get from the discrete Poincar\'e inequality~\eqref{eq:discretePoincareieqlity3} that
\begin{equation}\label{eq:20201104}
\|\boldsymbol E_h\|_0\lesssim \|\curl\boldsymbol E_h\|_0+\|\sigma_h\|_0 + \alpha.
\end{equation}
On the other side, by the fact $a(\boldsymbol E_h, \nabla^2\sigma_h)=0$ we have
\[
A_h(\boldsymbol E_h, \sigma_h;\boldsymbol E_h+\nabla^2\sigma_h, -\sigma_h)=a(\boldsymbol E_h, \boldsymbol E_h)+b_h(\nabla^2\sigma_h,\nabla^2\sigma_h)+c_h(\sigma_h,\sigma_h),
\]
which combined with~\eqref{eq:normequiv3}-\eqref{eq:normequiv4} implies
\[
\|\sigma_h\|_2^2+\|\curl\boldsymbol E_h\|_0^2\lesssim \alpha(\|\sigma_h\|_2 + \|\boldsymbol E_h+\nabla^2\sigma_h\|_{H(\curl)})\lesssim \alpha(\|\sigma_h\|_2 + \|\boldsymbol E_h\|_{H(\curl)}).
\]
Hence
\[
\|\sigma_h\|_2^2+\|\curl\boldsymbol E_h\|_0^2\lesssim \alpha^2+\alpha\|\boldsymbol E_h\|_0.
\]
Finally combining the last inequality and~\eqref{eq:20201104} gives~\eqref{eq:discreteinfsup}.
\end{proof}

From now on we always denote by $\boldsymbol E_h\in\boldsymbol\Sigma_h$ and $\sigma_h\in W_h$ the solution of the mixed method~\eqref{mvem1}-\eqref{mvem2}.
\begin{lemma}
 Assume $\boldsymbol E\in\boldsymbol H^{k+1}(\Omega;\mathbb S)$ and $\sigma\in\boldsymbol H^{k+2}(\Omega)$. Then
 \begin{align}\label{eq:consistencyerror1}
 b_h(\boldsymbol v_h,\nabla^2I_h^{\Delta}\sigma)- b(\boldsymbol v_h,\nabla^2\sigma) &\lesssim h^k\|\boldsymbol v_h\|_{0}|\sigma|_{k+2},\\
\label{eq:consistencyerror2}
 b_h(\boldsymbol I_h^c\boldsymbol E, \nabla^2\tau_h)- b(\boldsymbol E,\nabla^2\tau_h) &\lesssim h^{k+1}|\boldsymbol E|_{k+1}|\tau_h|_{2},\\
\label{eq:consistencyerror3}
 (\sigma,\tau_h)- c_h(I_h^{\Delta}\sigma, \tau_h) &\lesssim h_K^{k+1}\|\sigma\|_{k+1,K}\|\tau_h\|_{2,K}.
\end{align}
\end{lemma}
\begin{proof}
For each $K\in\mathcal T_h$,  we acquire from~\eqref{eq:bKconsistency},~\eqref{eq:bKcontinuity} and~\eqref{eq:IhDeltaerrestimate} that
\begin{align*}
&\quad b_K(\boldsymbol v_h,\nabla^2I_h^{\Delta}\sigma)-(\boldsymbol v_h,\nabla^2\sigma)_K \\
&=b_K(\boldsymbol v_h, \nabla^2I_h^{\Delta}\sigma-\boldsymbol Q_k^K(\nabla^2\sigma))-(\boldsymbol v_h, \nabla^2\sigma-\boldsymbol Q_k^K(\nabla^2\sigma))_K \\
&\lesssim \|\boldsymbol v_h\|_{0,K}\|\nabla^2I_h^{\Delta}\sigma-\boldsymbol Q_k^K(\nabla^2\sigma)\|_{0,K} + \|\boldsymbol v_h\|_{0,K}\|\nabla^2\sigma-\boldsymbol Q_k^K(\nabla^2\sigma)\|_{0,K} \\
&\lesssim \|\boldsymbol v_h\|_{0,K}(|\sigma-I_h^{\Delta}\sigma|_{2,K} + \|\nabla^2\sigma-\boldsymbol Q_k^K(\nabla^2\sigma)\|_{0,K})\lesssim h_K^k\|\boldsymbol v_h\|_{0,K}|\sigma|_{k+2,K}.
\end{align*}
Thus
\begin{align*}
 b_h(\boldsymbol v_h,\nabla^2I_h^{\Delta}\sigma_h)- b(\boldsymbol v_h,\nabla^2\sigma)&=\sum_{K\in\mathcal T_h}\left(b_K(\boldsymbol v_h,\nabla^2I_h^{\Delta}\sigma_h)-(\boldsymbol v_h,\nabla^2\sigma)_K\right) \\
 & \lesssim h^k\|\boldsymbol v_h\|_{0}|\sigma|_{k+2},
\end{align*}
i.e.~\eqref{eq:consistencyerror1}.

Similarly it holds from~\eqref{eq:bKconsistency},~\eqref{eq:bKcontinuity} and~\eqref{eq:Ikcprop1} that
\begin{align*}
 b_K(\boldsymbol I_h^c\boldsymbol E, \nabla^2\tau_h)- (\boldsymbol E,\nabla^2\tau_h)_K&= b_K(\boldsymbol I_h^c\boldsymbol E-\boldsymbol Q_k^K\boldsymbol E, \nabla^2\tau_h)- (\boldsymbol E-\boldsymbol Q_k^K\boldsymbol E,\nabla^2\tau_h)_K \\
 &\lesssim (\|\boldsymbol I_h^c\boldsymbol E-\boldsymbol Q_k^K\boldsymbol E\|_{0,K}+\|\boldsymbol E-\boldsymbol Q_k^K\boldsymbol E\|_{0,K})|\tau_h|_{2,K} \\
 &\lesssim h_K^{k+1}|\boldsymbol E|_{k+1,K}|\tau_h|_{2,K}, 
\end{align*}
which yields~\eqref{eq:consistencyerror2}.

Employing~\eqref{eq:cKconsistency},~\eqref{eq:cKcontinuity},~\eqref{eq:Pikerrestimate} and~\eqref{eq:IhDeltaerrestimate}, we get
\begin{align*}
&\quad\; (\sigma,\tau_h)_K- c_K(I_h^{\Delta}\sigma, \tau_h)\\
&= (\sigma-\widetilde{Q}_{k+2}^K\sigma,\tau_h)_K- c_K(I_h^{\Delta}\sigma-\widetilde{Q}_{k+2}^K\sigma, \tau_h) + (\widetilde{Q}_{k+2}^K\sigma, \tau_h-\widetilde{Q}_{k+2}^K\tau_h)_K\\
&= (\sigma-\widetilde{Q}_{k+2}^K\sigma,\tau_h)_K- c_K(I_h^{\Delta}\sigma-\widetilde{Q}_{k+2}^K\sigma, \tau_h) + (\Pi^K\sigma-Q_{k-2}^K\Pi^K\sigma, \tau_h-\widetilde{Q}_{k+2}^K\tau_h)_K\\
&\lesssim \left (\|\sigma-\widetilde{Q}_{k+2}^K\sigma\|_{0,K}+\|I_h^{\Delta}\sigma-\widetilde{Q}_{k+2}^K\sigma\|_{0,K}\right )\|\tau_h\|_{0,K} \\
&\quad\;+ \|\Pi^K\sigma-Q_{k-2}^K\Pi^K\sigma\|_{0,K}\|\tau_h-\widetilde{Q}_{k+2}^K\tau_h\|_{0,K} \\
 &\lesssim \left (\|\sigma-\Pi^K\sigma\|_{0,K}+\|\sigma-I_h^{\Delta}\sigma\|_{0,K} \right )\|\tau_h\|_{0,K}+h_K^2\|\Pi^K\sigma-Q_{k-2}^K\Pi^K\sigma\|_{0,K}|\tau_h|_{2,K} \\
 &\lesssim h_K^{k+1}\|\sigma\|_{k+1,K}\|\tau_h\|_{2,K}.
\end{align*}
Therefore~\eqref{eq:consistencyerror3} is true.
\end{proof}

\begin{theorem}\label{thm:mvemerrorestimate}
Let $\boldsymbol E_h\in\boldsymbol\Sigma_h$ and $\sigma_h\in W_h$ be the solution of the mixed method~\eqref{mvem1}-\eqref{mvem2} and let $\boldsymbol E$ and $\sigma$ be the solution of~\eqref{EBsystem2term1}-\eqref{EBsystem2term2}. 
Assume $\boldsymbol E\in\boldsymbol H^{k+1}(\Omega;\mathbb S)$, $\sigma\in H^{k+2}(\Omega)$ and $\boldsymbol f\in \boldsymbol H^k(\Omega;\mathbb S)$. We have
\[
\|\boldsymbol E-\boldsymbol E_h\|_{H(\curl)}+\|\sigma-\sigma_h\|_2\lesssim h^k \left (|\boldsymbol E|_{k+1}+\|\sigma\|_{k+2}+|\boldsymbol f|_k \right ).
\]
\end{theorem}
\begin{proof}
Take any $\boldsymbol v_h\in\boldsymbol\Sigma_h$ and $\tau_h\in W_h$.
We get from the variational formulation~\eqref{EBsystem2term1}-\eqref{EBsystem2term2},~\eqref{eq:Ikcprop1} and estimates
\eqref{eq:consistencyerror1}-\eqref{eq:consistencyerror3} that
\begin{align*}
&\quad A_h(\boldsymbol I_h^c\boldsymbol E, I_h^{\Delta}\sigma; \boldsymbol v_h, \tau_h)-(\boldsymbol f, \boldsymbol v_h) \\
&=a(\boldsymbol I_h^c\boldsymbol E-\boldsymbol E, \boldsymbol v_h)+b_h(\boldsymbol v_h,\nabla^2I_h^{\Delta}\sigma)- b(\boldsymbol v_h,\nabla^2\sigma) \\
&\quad+ b_h(\boldsymbol I_h^c\boldsymbol E, \nabla^2\tau_h)- b(\boldsymbol E,\nabla^2\tau_h)+(\sigma,\tau_h)- c_h(I_h^{\Delta}\sigma, \tau_h) \\
&\lesssim h^k|\boldsymbol E|_{k+1}\|\curl\boldsymbol v_h\|_0+h^k\|\boldsymbol v_h\|_{0}|\sigma|_{k+2}+h^{k+1}|\boldsymbol E|_{k+1}|\tau_h|_{2}+h^{k+1}\|\sigma\|_{k+1}\|\tau_h\|_{2}.
\end{align*}
Since 
\[
(\boldsymbol f, \boldsymbol v_h-\boldsymbol Q_h\boldsymbol v_h)=(\boldsymbol f-\boldsymbol Q_h\boldsymbol f, \boldsymbol v_h)\leq \|\boldsymbol f-\boldsymbol Q_h\boldsymbol f\|_0\|\boldsymbol v_h\|_0\lesssim h^k|\boldsymbol f|_k\|\boldsymbol v_h\|_0,
\]
we achieve from the mixed method~\eqref{mvem1}-\eqref{mvem2} that
\begin{align*}
&\quad A_h(\boldsymbol I_h^c\boldsymbol E-\boldsymbol E_h, I_h^{\Delta}\sigma-\sigma_h; \boldsymbol v_h, \tau_h)\\
&=A_h(\boldsymbol I_h^c\boldsymbol E, I_h^{\Delta}\sigma; \boldsymbol v_h, \tau_h)-(\boldsymbol f, \boldsymbol Q_h\boldsymbol v_h) \\
&=A_h(\boldsymbol I_h^c\boldsymbol E, I_h^{\Delta}\sigma; \boldsymbol v_h, \tau_h)-(\boldsymbol f, \boldsymbol v_h) + (\boldsymbol f, \boldsymbol v_h-\boldsymbol Q_h\boldsymbol v_h) \\
&\lesssim h^k|\boldsymbol E|_{k+1}\|\curl\boldsymbol v_h\|_0+h^k(|\sigma|_{k+2}+|\boldsymbol f|_k)\|\boldsymbol v_h\|_{0}+h^{k+1}|\boldsymbol E|_{k+1}|\tau_h|_{2} \\
&\quad+h^{k+1}\|\sigma\|_{k+1}\|\tau_h\|_{2}.
\end{align*}
Now it follows from the inf-sup condition~\eqref{eq:discreteinfsup} that \begin{align*}
\|\boldsymbol I_h^c\boldsymbol E-\boldsymbol E_h\|_{H(\curl)}+\|I_h^{\Delta}\sigma-\sigma_h\|_2&\lesssim \sup\limits_{\boldsymbol v_h\in\boldsymbol\Sigma_h\atop \tau_h\in W_h}\frac{A_h(\boldsymbol I_h^c\boldsymbol E-\boldsymbol E_h, I_h^{\Delta}\sigma-\sigma_h; \boldsymbol v_h, \tau_h)}{\|\boldsymbol v_h\|_{H(\curl)}+\|\tau_h\|_2} \\
&\lesssim h^k(|\boldsymbol E|_{k+1}+\|\sigma\|_{k+2}+|\boldsymbol f|_k).
\end{align*}
Thus we acquire from~\eqref{eq:Ikcprop1} and~\eqref{eq:IhDeltaerrestimate} that
\begin{align*}
&\quad\|\boldsymbol E-\boldsymbol E_h\|_{H(\curl)}+\|\sigma-\sigma_h\|_2\\
&\leq \|\boldsymbol E-\boldsymbol I_h^c\boldsymbol E\|_{H(\curl)}+\|\sigma-I_h^{\Delta}\sigma\|_2+ \|\boldsymbol I_h^c\boldsymbol E-\boldsymbol E_h\|_{H(\curl)}+\|I_h^{\Delta}\sigma-\sigma_h\|_2 \\
&\lesssim h^k(|\boldsymbol E|_{k+1}+\|\sigma\|_{k+2}+|\boldsymbol f|_k),
\end{align*}
as required.
\end{proof}


\end{document}